\documentclass[conference]{IEEEtran}

\usepackage[pdfencoding=auto]{hyperref}
\usepackage{balance}
\usepackage{cite}
\usepackage{amsmath, amssymb, amsfonts, amsthm}
\usepackage{enumitem}
\usepackage{graphicx}
\usepackage{textcomp}
\usepackage{xcolor}

\usepackage{multirow}
\usepackage{mathtools}
\usepackage{subfigure}

\usepackage[linesnumbered, ruled, vlined]{algorithm2e}
\SetKwInOut{Input}{Input}
\SetKwInOut{Output}{Output}

\graphicspath{{./figures/}}
\DeclareMathOperator*{\argmax}{arg\,max}
\DeclareMathOperator*{\argmin}{arg\,min}

\newtheorem{definition}{Definition}
\newtheorem{theorem}{Theorem}
\newtheorem{lemma}{Lemma}

\title{Streaming Algorithms for Diversity Maximization with Fairness Constraints}

\begin{document}

\author{
  \IEEEauthorblockN{Yanhao Wang\IEEEauthorrefmark{1}, Francesco Fabbri\IEEEauthorrefmark{2}\textsuperscript{,}\IEEEauthorrefmark{4}, Michael Mathioudakis\IEEEauthorrefmark{3}}
  \IEEEauthorblockA{\IEEEauthorrefmark{1}\textit{School of Data Science and Engineering, East China Normal University, Shanghai, China} \\
  \IEEEauthorrefmark{2}\textit{Department of Information and Communication Technologies, Universitat Pompeu Fabra, Barcelona, Spain} \\
  \IEEEauthorrefmark{4}\textit{Eurecat -- Centre Tecnol\`{o}gic de Catalunya, Barcelona, Spain} \\
  \IEEEauthorrefmark{3}\textit{Department of Computer Science, University of Helsinki, Helsinki, Finland} \\
  \IEEEauthorrefmark{1}yhwang@dase.ecnu.edu.cn\quad\IEEEauthorrefmark{2}francesco.fabbri@upf.edu\quad\IEEEauthorrefmark{3}michael.mathioudakis@helsinki.fi}
}

\maketitle

\begin{abstract}
Diversity maximization is a fundamental problem with wide applications in data summarization, web search, and recommender systems. Given a set $X$ of $n$ elements, it asks to select a subset $S$ of $k \ll n$ elements with maximum \emph{diversity}, as quantified by the dissimilarities among the elements in $S$. In this paper, we focus on the diversity maximization problem with fairness constraints in the streaming setting. Specifically, we consider the max-min diversity objective, which selects a subset $S$ that maximizes the minimum distance (dissimilarity) between any pair of distinct elements within it. Assuming that the set $X$ is partitioned into $m$ disjoint groups by some sensitive attribute, e.g., sex or race, ensuring \emph{fairness} requires that the selected subset $S$ contains $k_i$ elements from each group $i \in [1,m]$. A streaming algorithm should process $X$ sequentially in one pass and return a subset with maximum \emph{diversity} while guaranteeing the fairness constraint. Although diversity maximization has been extensively studied, the only known algorithms that can work with the max-min diversity objective and fairness constraints are very inefficient for data streams. Since diversity maximization is NP-hard in general, we propose two approximation algorithms for fair diversity maximization in data streams, the first of which is $\frac{1-\varepsilon}{4}$-approximate and specific for $m=2$, where $\varepsilon \in (0,1)$, and the second of which achieves a $\frac{1-\varepsilon}{3m+2}$-approximation for an arbitrary $m$. Experimental results on real-world and synthetic datasets show that both algorithms provide solutions of comparable quality to the state-of-the-art algorithms while running several orders of magnitude faster in the streaming setting.
\end{abstract}

\begin{IEEEkeywords}
algorithmic fairness, diversity maximization, max-min dispersion, streaming algorithm
\end{IEEEkeywords}

\section{Introduction}
\label{sec:intro}

\emph{Data summarization} is a common approach to tackling the challenges of big data volume in data-intensive applications. That is because, rather than performing high-complexity analysis tasks on the entire dataset, it is often beneficial to perform them on a representative and significantly smaller summary of the dataset, thus reducing the processing costs in terms of both running time and space usage. Typical examples of data summarization techniques~\cite{DBLP:journals/kais/Ahmed19} include sampling, sketching, histograms, coresets, and submodular maximization.

In this paper, we focus on \emph{diversity-aware data summarization}, which finds application in a wide range of real-world problems. For example, in database query processing~\cite{DBLP:journals/pvldb/QinYC12,DBLP:journals/kais/ZhengWQLG17}, web search~\cite{DBLP:conf/www/GollapudiS09,DBLP:conf/www/RafieiBS10}, and recommender systems~\cite{DBLP:journals/kbs/KunaverP17}, the output might be too large to be presented to the user entirely, even after filtering the results by relevance. One feasible solution is to present the user with a small but diverse subset that is easy for the user to process and representative of the complete results. As another example, when training machine learning models on massive data, feature and subset selection is a standard method to improve the efficiency. As indicated by several studies~\cite{DBLP:conf/aaai/ZadehGMZ17,DBLP:conf/icml/CelisKS0KV18}, selecting diverse features or subsets can lead to better balance between efficiency and accuracy. A key technical problem in both applications is \emph{diversity maximization}~\cite{DBLP:conf/pods/BorodinLY12,DBLP:journals/pvldb/DrosouP12,DBLP:conf/kdd/AbbassiMT13,DBLP:conf/pods/IndykMMM14,DBLP:conf/wsdm/CeccarelloPP18,DBLP:conf/pods/BorassiELVZ19,DBLP:conf/sdm/BauckhageSW20,DBLP:journals/pvldb/CeccarelloPPU17,moumoulidou_et_al:LIPIcs.ICDT.2021.13,DBLP:journals/tkde/DrosouP14,DBLP:conf/sdm/ZhangG20}.

In more detail, for a given set $X$ of elements in some metric space and a size constraint $k$, the diversity maximization problem asks for a subset of $k$ elements with maximum \emph{diversity}. Formally, diversity is quantified by a function that captures how well the selected subset spans the range of elements in $X$, and is typically defined in terms of distances or dissimilarities among elements in the subset. Prior studies~\cite{DBLP:conf/pods/IndykMMM14,DBLP:journals/kais/ZhengWQLG17,DBLP:conf/www/GollapudiS09,DBLP:journals/kbs/KunaverP17} have suggested many different diversity objectives of this kind. Two of the most popular ones are \emph{max-sum dispersion}, which aims to maximize the sum of distances between all pairs of elements in the selected subset $S$, and \emph{max-min dispersion}, which aims to maximize the minimum distance between any pair of distinct elements in $S$. Fig.~\ref{fig:compare} illustrates how to select 10 most diverse points from a point set in 2D with each of the two diversity objectives. As shown in Fig.~\ref{fig:compare}, the max-sum dispersion tends to select ``marginal'' elements and may include highly similar elements in the solution, which is not suitable for the applications requiring more uniform coverage. Therefore, we focus on diversity maximization based on the max-min dispersion problem in this paper.

\begin{figure}[t]
  \centering
  \subfigure[Max-sum dispersion]{
    \label{subfig:max:sum}
    \includegraphics[width=0.18\textwidth]{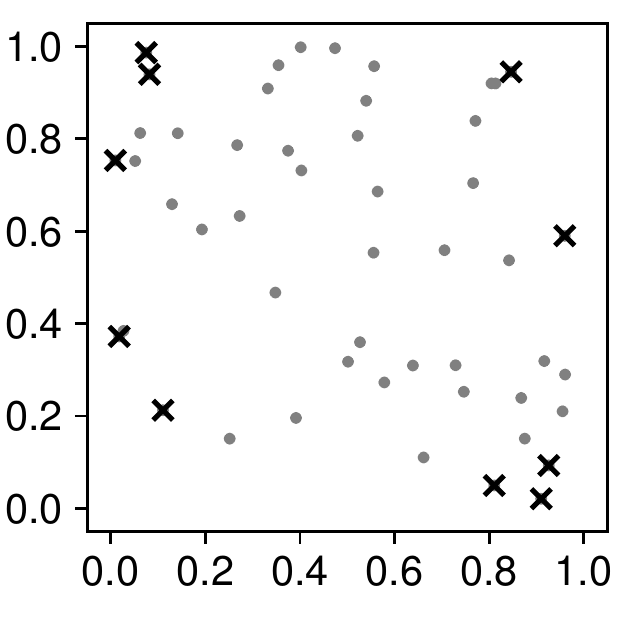}
  }
  \subfigure[Max-min dispersion]{
    \label{subfig:max:min}
    \includegraphics[width=0.18\textwidth]{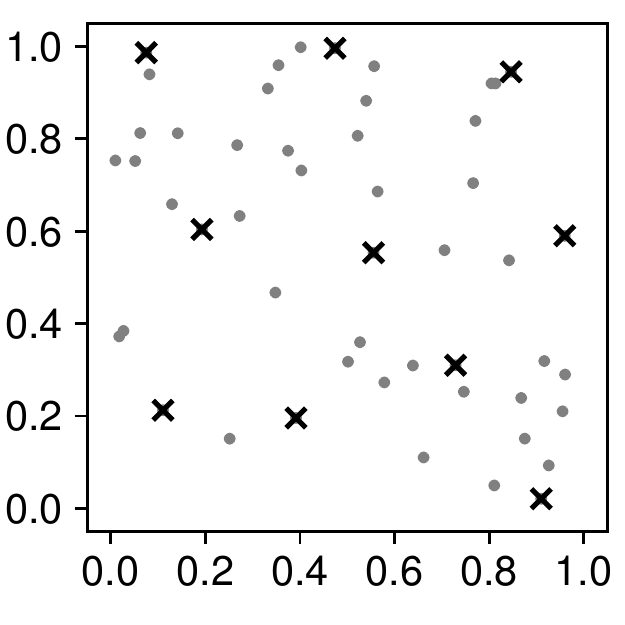}
  }
  \caption{Comparison of max-sum and max-min dispersion diversity objectives.}
  \label{fig:compare}
\end{figure}

In addition to \emph{diversity}, \emph{fairness} in data summarization is also attracting increasing attention~\cite{DBLP:conf/icml/CelisKS0KV18,DBLP:conf/icml/ChiplunkarKR20,DBLP:conf/icml/JonesNN20,DBLP:conf/icml/KleindessnerAM19,DBLP:conf/waoa/0001SS19,DBLP:conf/nips/HuangJV19,NEURIPS2020_9d752cb0,DBLP:conf/www/0001FM21}. Several studies reveal that the biases w.r.t.~sensitive attributes, e.g., sex, race, or age, in underlying datasets can be retained in the summaries and could lead to unfairness in data-driven socio-computational systems such as education, recruitment, and banking~\cite{DBLP:conf/icml/CelisKS0KV18,DBLP:conf/icml/KleindessnerAM19,NEURIPS2020_9d752cb0}. One of the most common notions for fairness in data summarization is \emph{group fairness}~\cite{DBLP:conf/icml/ChiplunkarKR20,DBLP:conf/icml/JonesNN20,DBLP:conf/icml/KleindessnerAM19,DBLP:conf/icml/CelisKS0KV18,DBLP:conf/www/0001FM21}, which partitions the dataset into $m$ disjoint groups based on some sensitive attribute and introduces a \emph{fairness constraint} that limits the number of elements from group $i$ in the summary to $k_i$ for every group $i \in [1,m]$. However, most existing methods for diversity maximization cannot be adapted directly to satisfy such fairness constraints. Moreover, a few methods that can deal with fairness constraints are specific for the max-sum dispersion problem~\cite{DBLP:conf/kdd/AbbassiMT13,DBLP:conf/pods/BorodinLY12,DBLP:conf/wsdm/CeccarelloPP18}. To the best of our knowledge, the methods in~\cite{moumoulidou_et_al:LIPIcs.ICDT.2021.13} are the only ones for max-min diversity maximization with fairness constraints.

Furthermore, since the applications of diversity maximization are mostly in the realm of massive data analysis, it is important to design efficient algorithms for processing large-scale datasets. The streaming model is a well-recognized framework for big data processing. In the streaming model, an algorithm is only permitted to process each element in the dataset sequentially in one pass, is allowed to take time and space that are sublinear to or even independent of the dataset size, and is required to provide solutions of nearly equal quality to those returned by offline algorithms. However, the only known algorithms~\cite{moumoulidou_et_al:LIPIcs.ICDT.2021.13} for fair max-min diversity maximization are designed for the offline setting and are thus very inefficient in data streams.

\begin{figure}[t]
\centering
\subfigure[Unconstrained solution]{
  \label{subfig:unconstrained:dm}
  \includegraphics[width=0.18\textwidth]{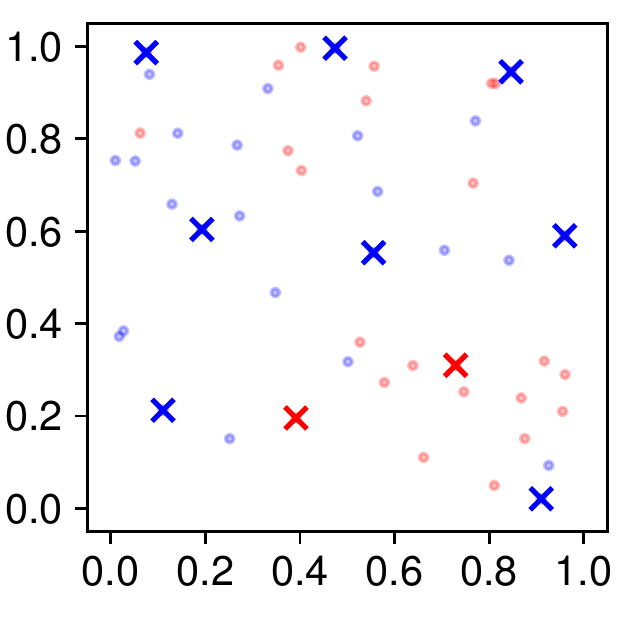}
}
\subfigure[Fair solution]{
  \label{subfig:fair:dm}
  \includegraphics[width=0.18\textwidth]{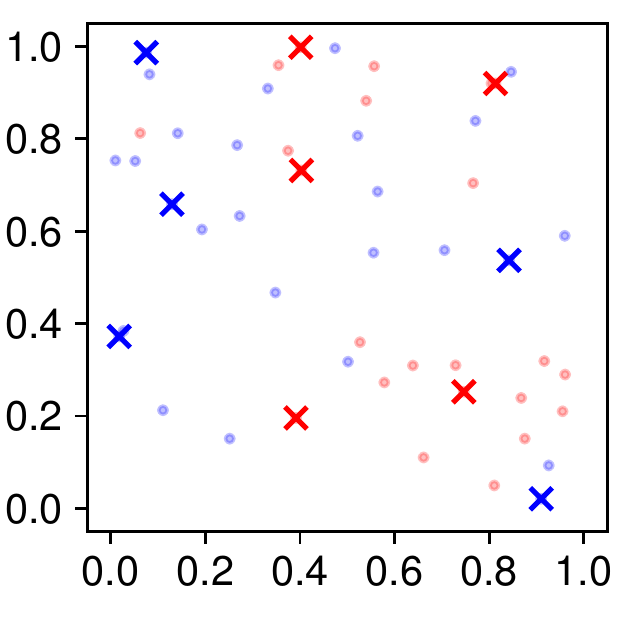}
}
\caption{Example of unconstrained vs.~fair diversity maximization. We consider a set of individuals each of whom is described by two attributes such as \emph{income} and \emph{capital gain}. They are partitioned into two disjoint groups by \emph{gender}, as colored in red and blue, respectively. The fair diversity maximization problem returns a subset of size $10$ that maximizes diversity in terms of attributes and contains an equal number (i.e., $k_i = 5$) of elements from both groups.}
\label{fig:fair:dm}
\end{figure}

\vspace{1mm}
\noindent\textbf{Our Contributions:}
In this paper we propose novel streaming algorithms for the fair diversity maximization (FDM) problem with the max-min objective.
Our main contributions are summarized as follows.
\begin{itemize}[leftmargin=*]
  \item We formally define the fair max-min diversity maximization problem in metric spaces. Then, we describe the streaming algorithm for unconstrained max-min diversity maximization in~\cite{DBLP:conf/pods/BorassiELVZ19} and improve its approximation ratio from $\frac{1-\varepsilon}{5}$ to $\frac{1-\varepsilon}{2}$ for any parameter $\varepsilon\in(0,1)$.
  \item We first propose a $\frac{1-\varepsilon}{4}$-approximation streaming algorithm called \textsf{SFDM1} for FDM when there are only two groups in the dataset. During stream processing, it maintains group-blind solutions and group-specific solutions for both groups using the streaming algorithm of~\cite{DBLP:conf/pods/BorassiELVZ19}. In the post-processing, each group-blind solution is balanced for the fairness constraint by swapping elements with group-specific solutions. It takes $O(\frac{k\log\Delta}{\varepsilon})$ time per element for streaming processing, where $\Delta$ is the ratio of the maximum and minimum distances between any pair of distinct elements in the dataset, spends $O(\frac{k^2\log\Delta}{\varepsilon})$ time for post-processing, and stores $O(\frac{k\log\Delta}{\varepsilon})$ elements in memory.
  \item We further propose a $\frac{1-\varepsilon}{3m+2}$-approximation streaming algorithm called \textsf{SFDM2} for FDM with an arbitrary number $m$ of groups in the dataset. It uses a similar method for stream processing to \textsf{SFDM1}. In the post-processing, it first partitions the elements in all solutions into clusters based on their pairwise distances. Starting from a partial solution chosen from the group-blind solution, it utilizes an algorithm to find a maximum-cardinality intersection of two matroids, the first of which is defined by the fairness constraint and the second of which is defined on the clusters, to augment the partial solution to acquire the final fair solution. It takes $O(\frac{k\log\Delta}{\varepsilon})$ time per element for streaming processing, requires $ O\big(\frac{k^2m\log{\Delta}}{\varepsilon} \cdot (m+\log^{2}{k})\big) $ time for post-processing, and stores $O(\frac{km\log\Delta}{\varepsilon})$ elements in memory.
  \item Finally, we evaluate the performance of our proposed algorithms against the state-of-the-art algorithms on several real-world and synthetic datasets. The results demonstrate that our proposed algorithms provide solutions of comparable quality to those returned by the state-of-the-art algorithms while running several orders of magnitude faster than them in the streaming setting.
\end{itemize}

The rest of this paper is organized as follows. The related work is discussed in Section~\ref{sec:related:work}. In Section~\ref{sec:def}, we introduce the basic concepts and definitions and describe the streaming algorithm for unconstrained diversity maximization. In Section~\ref{sec:alg}, we present our streaming algorithms for fair diversity maximization. Our experimental setup and results are reported in Section~\ref{sec:exp}. Finally, we conclude this paper in Section~\ref{sec:conclusion}.

\section{Related Work}
\label{sec:related:work}

Diversity maximization has been extensively studied over the last two decades. Existing studies mostly focus on two popular diversity objectives -- i.e., \emph{max-sum dispersion}~\cite{DBLP:journals/ior/RaviRT94,DBLP:journals/orl/HassinRT97,DBLP:conf/kdd/AbbassiMT13,DBLP:conf/pods/IndykMMM14,DBLP:conf/compgeom/CevallosEZ16,DBLP:conf/soda/CevallosEZ17,DBLP:conf/wsdm/CeccarelloPP18,DBLP:conf/pods/BorassiELVZ19,DBLP:conf/sdm/BauckhageSW20,DBLP:conf/cccg/AghamolaeiFZ15,DBLP:journals/pvldb/CeccarelloPPU17} and \emph{max-min dispersion}~\cite{DBLP:journals/ior/RaviRT94,DBLP:journals/orl/HassinRT97,DBLP:conf/pods/IndykMMM14,DBLP:conf/pods/BorassiELVZ19,moumoulidou_et_al:LIPIcs.ICDT.2021.13,DBLP:journals/tkde/DrosouP14,DBLP:conf/cccg/AghamolaeiFZ15,DBLP:journals/pvldb/CeccarelloPPU17}, as well as their variants~\cite{DBLP:journals/jal/ChandraH01,DBLP:conf/pods/IndykMMM14}.

An early study~\cite{ERKUT199048} proved that the max-sum and max-min diversity maximization problems are NP-hard even in metric spaces. The classic approaches to both problems are the greedy algorithms~\cite{DBLP:journals/ior/RaviRT94,DBLP:journals/orl/HassinRT97}, which achieves the best possible approximation ratio of $\frac{1}{2}$ unless P=NP. Indyk et al.~\cite{DBLP:conf/pods/IndykMMM14} proposed composable coreset-based approximation algorithms for diversity maximization. Aghamolaei et al.~\cite{DBLP:conf/cccg/AghamolaeiFZ15} improved the approximation ratios in~\cite{DBLP:conf/pods/IndykMMM14}. Ceccarello et al.~\cite{DBLP:journals/pvldb/CeccarelloPPU17} proposed coreset-based approximation algorithms for diversity maximization in MapReduce and streaming settings where the metric space has a bounded doubling dimension. Borassi et al.~\cite{DBLP:conf/pods/BorassiELVZ19} proposed sliding-window streaming algorithms for diversity maximization. Drosou and Pitoura~\cite{DBLP:journals/tkde/DrosouP14} studied max-min diversity maximization on dynamic data. They proposed a $\frac{b-1}{2b^2}$-approximation algorithm using a cover tree of base $b$. Bauckhage et al.~\cite{DBLP:conf/sdm/BauckhageSW20} proposed an adiabatic quantum computing solution for max-sum diversification. Zhang and Gionis~\cite{DBLP:conf/sdm/ZhangG20} extended diversity maximization to clustered data. Nevertheless, all the above methods only consider diversity maximization problems without fairness constraints.

There have been several studies on diversity maximization under matroid constraints, of which the fairness constraints are special cases. Abbassi et al.~\cite{DBLP:conf/kdd/AbbassiMT13} proposed a $(\frac{1}{2}-\varepsilon)$-approximation local search algorithm for max-sum diversification under matroid constraints. Borodin et al.~\cite{DBLP:conf/pods/BorodinLY12} proposed a $(\frac{1}{2}-\varepsilon)$-approximation algorithm for maximizing the sum of a submodular function and a max-sum dispersion function. Cevallos et al.~\cite{DBLP:conf/soda/CevallosEZ17} extended the local search algorithm for distances of negative type. They also proposed a PTAS for this problem via convex programming~\cite{DBLP:conf/compgeom/CevallosEZ16}. Bhaskara et al.~\cite{DBLP:conf/nips/BhaskaraGMS16} proposed a $\frac{1}{8}$-approximation algorithm for sum-min diversity maximization using linear relaxations. Ceccarello et al.~\cite{DBLP:conf/wsdm/CeccarelloPP18} proposed a coreset-based approach to matroid-constrained max-sum diversification in metric spaces of bounded doubling dimension. Nevertheless, the above methods are still not applicable to the max-min dispersion objective. The only known algorithms for fair max-min diversity maximization in~\cite{moumoulidou_et_al:LIPIcs.ICDT.2021.13} are offline algorithms and inefficient for data streams. We will compare our proposed algorithms with them both theoretically and experimentally in this paper. To the best of our knowledge, there has not been any previous steaming algorithm for fair max-min diversity maximization.

In addition to diversity maximization, \emph{fairness} has also been considered in many other data summarization problems, such as $k$-center~\cite{DBLP:conf/icml/ChiplunkarKR20,DBLP:conf/icml/JonesNN20,DBLP:conf/icml/KleindessnerAM19}, determinantal point processes~\cite{DBLP:conf/icml/CelisKS0KV18}, coresets for $k$-means clustering~\cite{DBLP:conf/waoa/0001SS19,DBLP:conf/nips/HuangJV19}, and submodular maximization~\cite{NEURIPS2020_9d752cb0,DBLP:conf/www/0001FM21}. However, since their optimization objectives are different from diversity maximization, the proposed methods cannot be directly used for our problem.

\section{Preliminaries}
\label{sec:def}

\subsection{Fair Diversity Maximization}
\label{subsec:def}

Let $X$ be a set of $n$ elements from a metric space with distance function $d(\cdot,\cdot)$ capturing the dissimilarities among elements. Recall that $d(\cdot,\cdot)$ is \emph{nonnegative}, \emph{symmetric}, and satisfies the \emph{triangle inequality} -- i.e., $ d(x, y) + d(y, z) \geq d(x, z) $ for any $x, y, z \in X$. Note that all the algorithms and analyses in this paper are general for any distance metric. We further generalize the notion of \emph{distance} to an element $x$ and a set $S$ as the distance between $x$ and its nearest neighbor in $S$ -- i.e., $ d(x,S) = \min_{y \in S} d(x,y) $.

Our focus in this paper is to find a small subset of \emph{most diverse} elements from $X$. Given a subset $S \subseteq X$, its diversity $div(S)$ is defined as the minimum of the pairwise distances between any two distinct elements in $S$ -- i.e., $div(S) = \min_{x,y \in S, x \neq y} d(x,y)$. The unconstrained version of diversity maximization (DM) asks for a subset $S \subseteq X$ of $k$ elements maximizing $div(S)$ -- i.e., $S^* = \argmax_{S \subseteq X : |S| = k} div(S)$. We use $ \mathtt{OPT}=div(S^*)$ to denote the diversity of the optimal solution $S^*$ for DM. This problem has been proven to be NP-complete~\cite{ERKUT199048}, and no polynomial-time algorithm can achieve an approximation factor of better than $\frac{1}{2}$ unless P=NP. One approach to DM is the $\frac{1}{2}$-approximation greedy algorithm~\cite{DBLP:journals/tcs/Gonzalez85,DBLP:journals/ior/RaviRT94} (known as \textsf{GMM}) in the offline setting.

We introduce \emph{fairness} to diversity maximization in case that $X$ is comprised of demographic groups defined on some sensitive attribute, e.g., sex or race. Formally, suppose that $X$ is divided into $m$ disjoint groups $[1,\dots,m]$ ($[m]$ for short) and a function $c : X \rightarrow [m]$ maps each element $x \in X$ to the group it belongs to. Let $ X_i = \{ x \in X : c(x)=i \} $ be the subset of elements from group $i$ in $X$. Obviously, we have $ \bigcup_{i=1}^{m} X_i = X $ and $ X_i \cap X_j = \emptyset $ for any $ i \neq j $. The fairness constraint assigns a positive integer $k_i$ to each of the $m$ groups and restricts the number of elements from group $i$ in the solution to $k_i$. We assume that $ \sum_{i=1}^{m} k_i = k $. The fair diversity maximization (FDM) problem is defined as follows.
\begin{definition}[Fair Diversity Maximization]
  Given a set $X$ of $n$ elements with $m$ disjoint groups $X_1,\ldots,X_m$ and $m$ size constraints $k_1,\ldots,k_m \in \mathbb{Z}^{+}$, find a subset $S$ that contains $k_i$ elements from $X_i$ and maximizes $div(S)$ -- i.e.,
  \begin{math}
    S^*_{f} = \argmax_{S \subseteq X \, : \, |S \cap X_i| = k_i, \forall i \in [m]} div(S)
  \end{math}.
\end{definition}
We use $ \mathtt{OPT}_f = div(S^*_{f}) $ to denote the diversity of the optimal solution $S^*_{f}$ for FDM. Since DM is a special case of FDM when $m=1$, FDM is NP-hard up to a $\frac{1}{2}$-approximation as well. In addition, our FDM problem is closely related to the concept of \emph{matroid}~\cite{Korte2012} in combinatorics. Given a ground set $V$, a matroid is a pair $\mathcal{M}=(V,\mathcal{I})$ where $\mathcal{I}$ is a family of subsets of $V$ (called the independent sets) with the following properties: (1) $\emptyset \in \mathcal{I}$; (2) for each $A \subseteq B \subseteq V$, if $B \in \mathcal{I}$ then $A \in \mathcal{I}$ (\emph{hereditary property}); and (3) if $A \in \mathcal{I}$, $B \in \mathcal{I}$, and $|A| > |B|$, then there exists $x \in A \setminus B$ such that $B \cup \{x\} \in \mathcal{I}$ (\emph{augmentation property}). An independent set is maximal if it is not a proper subset of any other independent set. A basic property of $\mathcal{M}$ is that its all maximal independent sets have the same size, which is denoted as the rank of the matroid. As is easy to verify, our fairness constraint is a case of rank-$k$ \emph{partition matroids}, where the ground set is partitioned into disjoint groups and the independent sets are exactly the sets in which, for each group, the number of elements from this group is at most the group capacity. And our streaming algorithms in Section~\ref{sec:alg} will be built on the properties of matroids. In this paper, we study the FDM problem in the streaming setting, where the elements in $X$ arrive one at a time and an algorithm must process each element sequentially in one pass using limited space (typically independent of $n$) and return a valid approximate solution $S$ for FDM.

\subsection{Streaming Algorithm for Diversity Maximization}
\label{subsec:alg:stream}

Before presenting our proposed algorithms for FDM, we first recall the streaming algorithm for (unconstrained) DM proposed by Borassi et al.~\cite{DBLP:conf/pods/BorassiELVZ19} in Algorithm~\ref{alg:sdm}. Let $ d_{min} = \min_{x,y \in X, x \neq y} d(x,y) $, $ d_{max} = \max_{x,y \in X} d(x,y) $, and $ \Delta = \frac{d_{max}}{d_{min}} $. Obviously, it always holds that $ \mathtt{OPT} \in [d_{min},d_{max}] $. Algorithm~\ref{alg:sdm} maintains a sequence $\mathcal{U}$ of values for guessing $\mathtt{OPT}$ within relative errors of $1-\varepsilon$ and initializes an empty solution $S_{\mu}$ for each $ \mu \in \mathcal{U} $ before processing the stream. Then, for each $x \in X$ and each $\mu \in \mathcal{U}$, if $ S_{\mu} $ contains less than $k$ elements and the distance between $x$ and $S_{\mu}$ is at least $\mu$, it will add $x$ to $S_{\mu}$. After processing all elements in $X$, the candidate solution that contains $k$ elements and maximizes the diversity is returned as the solution $S$ for DM.

\begin{algorithm}[t]
  \caption{Streaming Diversity Maximization}\label{alg:sdm}
  \Input{A stream $X$, a distance metric $d$, a parameter $\varepsilon \in (0,1)$, a size constraint $k \in \mathbb{Z}^{+}$}
  \Output{A set $S \subseteq X$ with $|S|=k$}
  $\mathcal{U} = \{\frac{d_{min}}{(1-\varepsilon)^j} \, : \, j \in \mathbb{Z}_{0}^{+} \wedge \frac{d_{min}}{(1-\varepsilon)^j} \in [d_{min}, d_{max}] \}$\;
  Initialize $ S_{\mu} = \emptyset $ for each $ \mu \in \mathcal{U} $\;
  \ForEach{$x \in X$}{
    \ForEach{$ \mu \in \mathcal{U} $}{
      \If{$ |S_{\mu}| < k $ \textnormal{and} $ d(x,S_{\mu}) \geq \mu $\label{ln:sdm:cond}}{
        $S_{\mu} \gets S_{\mu} \cup \{x\}$\;
      }
    }
  }
  \Return{$S \gets \argmax_{\mu \in \mathcal{U} : |S_{\mu}|=k} div(S_{\mu}) $}\;
\end{algorithm}

Algorithm~\ref{alg:sdm} has been proven to be $\frac{1-\varepsilon}{5}$-approximate for many diversity objectives~\cite{DBLP:conf/pods/BorassiELVZ19} including max-min dispersion. In Theorem~\ref{thm:sdm:approx}, we improve its approximation ratio for max-min dispersion to $\frac{1-\varepsilon}{2}$ by refining the analysis of~\cite{DBLP:conf/pods/BorassiELVZ19}. 

\begin{theorem}\label{thm:sdm:approx}
  Algorithm~\ref{alg:sdm} is a $\frac{1-\varepsilon}{2}$-approximation algorithm for the max-min diversity maximization problem.
\end{theorem}
\begin{proof}
  For each $\mu \in \mathcal{U}$, there are two cases for $ S_{\mu} $ after processing all elements in $X$: (1) If $|S_{\mu}| = k$, the condition of Line~\ref{ln:sdm:cond} guarantees that $ div(S_{\mu}) \geq \mu $; (2) If $|S_{\mu}| < k$, it holds that $ d(x,S_{\mu}) < \mu $ for every $ x \in X \setminus S_{\mu} $ since the fact that $x$ is not added to $S_{\mu}$ implies that $ d(x,S_{\mu}) < \mu $, as $ |S_{\mu}| < k $. Let us consider some $S_{\mu}$ with $|S_{\mu}| < k$. Suppose that $ S^*=\{s_1^*,\ldots,s_k^*\} $ is the optimal solution for DM on $X$. We define a function $f: S^* \rightarrow S_{\mu} $ that maps each element in $S^*$ to its nearest neighbor in $S_{\mu}$. As is shown above, $ d(s^*,f(s^*)) < \mu $ for each $s^* \in S^*$. Because $|S_{\mu}| < k$ and $ |S^*|=k $, there must exist two distinct elements $ s_a^*,s_b^* \in S^* $ with $ f(s_a^*) = f(s_b^*) $. For such $s_a^*,s_b^*$, we have
  \begin{displaymath}
    d(s_a^*,s_b^*) \leq d(s_a^*,f(s_a^*)) + d(s_b^*,f(s_b^*)) < 2\mu
  \end{displaymath}
  according to the triangle inequality. Thus, $ \mathtt{OPT} = div(S^*) \leq d(s_a^*,s_b^*) < 2\mu $ if $|S_{\mu}| < k$. Let $\mu^{\prime}$ be the smallest $ \mu \in \mathcal{U} $ with $ |S_{\mu}| < k $. We have got $ div(S^*) < 2\mu^{\prime} $ from the above results. Additionally, for $\mu^{\prime\prime} = (1-\varepsilon)\mu^{\prime}$, we must have $ |S_{\mu^{\prime\prime}}| = k $ and $div(S_{\mu^{\prime\prime}}) \geq \mu^{\prime\prime}$. Therefore, we have $ div(S) \geq \mu^{\prime\prime} = (1-\varepsilon)\mu^{\prime} \geq \frac{1-\varepsilon}{2} \cdot div(S^*) $ and conclude the proof.
\end{proof}

In terms of space and time complexity, Algorithm~\ref{alg:sdm} stores $ O(\frac{k\log{\Delta}}{\varepsilon}) $ elements and takes $ O(\frac{k\log{\Delta}}{\varepsilon}) $ time per element since it makes $ O(\frac{\log{\Delta}}{\varepsilon}) $ guesses for $\mathtt{OPT}$, keeps at most $k$ elements in each candidate, and requires at most $k$ distance computations to determine whether to add an element to a candidate or not. In Section~\ref{sec:alg}, we will show how Algorithm~\ref{alg:sdm} serves a building block in our proposed algorithms for FDM.

\section{Algorithms}
\label{sec:alg}

As shown in Section~\ref{subsec:def}, the fair diversity maximization (FDM) problem is NP-hard in general. Therefore, we will focus on efficient approximation algorithms for this problem. We first propose a $\frac{1-\varepsilon}{4}$-approximation streaming algorithm for FDM in the special case that there are only $m=2$ groups in the dataset. Then, we propose a $\frac{1-\varepsilon}{3m+2}$-approximation streaming algorithm for an arbitrary number $m$ of groups.

\subsection{Streaming Algorithm for \texorpdfstring{$m=2$}{m=2}}
\label{subsec:alg:stream:1}

\begin{algorithm}[t]
  \caption{\textsf{SFDM1}}\label{alg:sfdm:1}
  \Input{A stream $X = X_1 \cup X_2$, a distance metric $d$, a parameter $\varepsilon \in (0,1)$, two size constraints $k_1,k_2 \in \mathbb{Z}^{+}$ ($ k=k_1+k_2 $)}
  \Output{A set $S \subseteq X$ s.t.~$ |S \cap X_i|=k_i $ for $i \in \{1,2\}$}
  \tcc{Stream processing}
  $\mathcal{U} = \{\frac{d_{min}}{(1-\varepsilon)^j} \, : \, j \in \mathbb{Z}_{0}^{+} \wedge \frac{d_{min}}{(1-\varepsilon)^j} \in [d_{min}, d_{max}] \}$\;
  Initialize $ S_{\mu},S_{\mu,i} = \emptyset $ for every $ \mu \in \mathcal{U} $ and $i \in \{1,2\}$\;
  \ForEach{$x \in X$}{
    \ForEach{$ \mu \in \mathcal{U} $ \textnormal{and} $ i \in \{1,2\} $}{
      \If{$ |S_{\mu}| < k \wedge d(x,S_{\mu}) \geq \mu $\label{ln:sfdm:group:blind}}{
        $S_{\mu} \gets S_{\mu} \cup \{x\}$\;
      }
      \If{$ c(x) = i \wedge |S_{\mu,i}| < k_i \wedge d(x,S_{\mu,i}) \geq \mu $\label{ln:sfdm:group:specific}}{
        $S_{\mu,i} \gets S_{\mu,i} \cup \{x\}$\;
      }
    }
  }
  \tcc{Post-processing}
  $\mathcal{U}^{\prime} = \{ \mu \in \mathcal{U} : |S_{\mu}|=k \wedge |S_{\mu,i}|=k_i, \forall i \in \{1,2\} \}$\;\label{ln:sfdm:mu}
  \ForEach{$\mu \in \mathcal{U}^{\prime}$}{
    \If{$|S_{\mu} \cap X_i| < k_i$ for some $i \in \{1,2\}$}{
      \While{$ |S_{\mu} \cap X_i| < k_i$}{
        $ x^{+} \gets \argmax_{x \in S_{\mu,i}} d(x,S_{\mu} \cap X_i) $\;\label{ln:sfdm:post:3}
        $ S_{\mu} \gets S_{\mu} \cup \{x^{+}\} $\;
      }
      \While{$ |S_{\mu}| > k $}{
        $ x^{-} \gets \argmin_{x \in S_{\mu} \setminus X_{i}} d(x,S_{\mu} \cap X_i) $\;\label{ln:sfdm:post:2}
        $ S_{\mu} \gets S_{\mu} \setminus \{x^{-}\} $\;
      }
    }
  }
  \Return{$S \gets \argmax_{\mu \in \mathcal{U}^{\prime}} div(S_{\mu})$}\;
\end{algorithm}

Now we present our streaming algorithm in case of $m=2$ called \textsf{SFDM1}. The detailed procedure of \textsf{SFDM1} is described in Algorithm~\ref{alg:sfdm:1}. In general, it runs in two phases: \emph{stream processing} and \emph{post-processing}. In stream processing, for each guess $\mu \in \mathcal{U}$ of $\mathtt{OPT}_f$, it utilizes Algorithm~\ref{alg:sdm} to keep a group-blind candidate $S_{\mu}$ with size constraint $k$ and two group-specific candidates $S_{\mu,1}$ and $S_{\mu,2}$ with size constraints $k_1$ and $k_2$ for $X_1$ and $X_2$, respectively. The only difference from Algorithm~\ref{alg:sdm} is that the elements are filtered by group for maintaining $S_{\mu,1}$ and $S_{\mu,2}$. After processing all elements of $X$ in one pass, it will post-process the group-blind candidates so as to make them satisfy the fairness constraint. The post-processing is only performed on a subset $\mathcal{U}^{\prime}$ of $\mathcal{U}$ where $S_{\mu}$ contains $k$ elements and $S_{\mu,i}$ contains $k_i$ elements for each group $i \in \{1,2\}$. For each $\mu \in \mathcal{U}^{\prime}$, $S_{\mu}$ either has satisfied the fairness constraint or has one over-filled group $i_o$ and the other under-filled group $i_u$. If $S_{\mu}$ is not yet a fair solution, $S_{\mu}$ will be balanced for fairness by first adding $k_{i_u}-k_{i_u}^{\prime}$ elements, where $k_{i_u}^{\prime}=|S_{\mu} \cap X_{i_u}|$, from $S_{\mu,i_u}$ to $S_{\mu}$ and then removing the same number of elements from $S_{\mu} \cap X_{i_o}$. The elements to be added and removed are selected greedily like \textsf{GMM}~\cite{DBLP:journals/tcs/Gonzalez85} to minimize the loss in diversity: the element in $S_{\mu,i_u}$ that is the furthest from $S_{\mu} \cap X_{i_u}$ is picked for each insertion; and the element in $S_{\mu} \cap X_{i_o}$ that is the closest to $S_{\mu} \cap X_{i_u}$ is picked for each deletion. Finally, the fair candidate with the maximum diversity after post-processing is returned as the final solution for FDM. Next, we will analyze the approximation ratio and complexity of \textsf{SFDM1} theoretically.

\begin{figure}[t]
  \centering
  \includegraphics[width=0.45\textwidth]{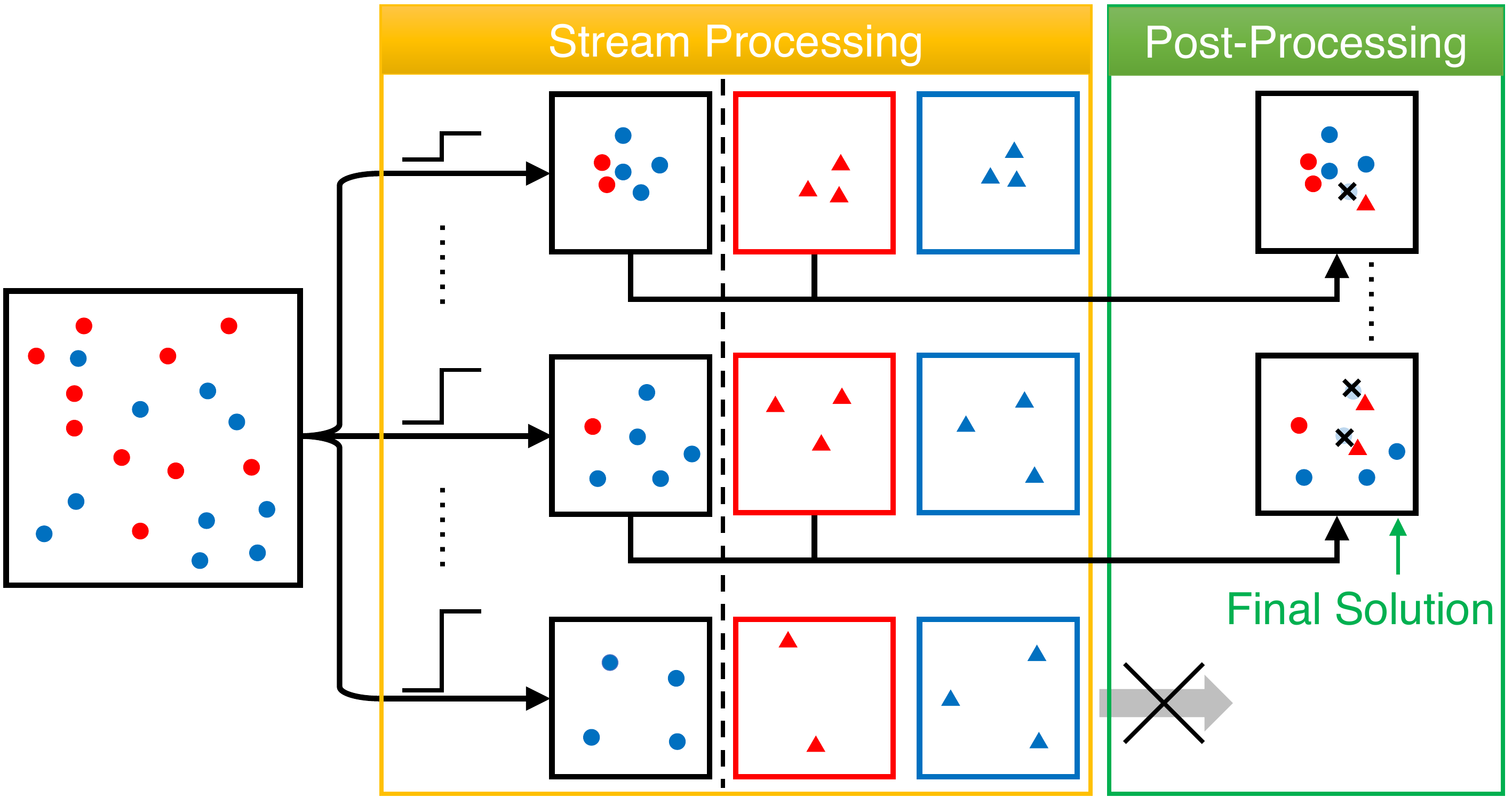}
  \caption{Illustration of \textsf{SFDM1}. In stream processing, one group-blind and two group-specific candidates are maintained for each guess $\mu$ of $\mathtt{OPT}_f$. Then, a subset of group-blind candidates are selected for post-processing by adding the elements from the under-filled group before deleting the elements from the over-filled one.}
  \label{fig:alg1}
\end{figure}

\vspace{1mm}
\noindent\textbf{Approximation Ratio:}
We first prove that \textsf{SFDM1} achieves an approximation ratio of $\frac{1-\varepsilon}{4}$ for FDM, where $\varepsilon \in (0,1)$. The proof is based on (1) the existence of $\mu^{\prime} \in \mathcal{U}^{\prime}$ such that $\mu^{\prime} \geq \frac{1-\varepsilon}{2} \cdot \mathtt{OPT}_f$ (Lemma~\ref{lm:sfdm:11}) and (2) $div(S_{\mu}) \geq \frac{\mu}{2}$ for each $\mu \in \mathcal{U}^{\prime}$ after post-processing (Lemma~\ref{lm:sfdm:12}).

\begin{lemma}\label{lm:sfdm:11}
  Let $\mu^{\prime}$ be the largest $\mu \in \mathcal{U}^{\prime}$. It holds that $ \mu^{\prime} \geq \frac{1-\varepsilon}{2} \cdot \mathtt{OPT}_f $ where $ \mathtt{OPT}_f $ is the optimal diversity of FDM on $X$.
\end{lemma}
\begin{proof}
  First of all, it is obvious that $ \mathtt{OPT}_f \leq \mathtt{OPT} $, where $ \mathtt{OPT} $ is the optimal diversity of unconstrained DM with $ k = k_1 + k_2 $ on $X$, since any valid solution for FDM must also be a valid solution for DM. Moreover, it holds that $ \mathtt{OPT}_f \leq \mathtt{OPT}_{k_i} $, where $ \mathtt{OPT}_{k_i} $ is the optimal diversity of unconstrained DM with size constraint $k_i$ on $X_i$ for both $ i \in \{1,2\} $, because the optimal solution must contain $k_i$ elements from $X_i$ and $div(\cdot)$ is a monotonically non-increasing function -- i.e., $ div(S \cup \{x\}) \leq div(S) $ for any $S \subseteq X$ and $x \in S \setminus X$. Therefore, we prove that $\mathtt{OPT}_f \leq div(S^* \cap X_i) \leq \mathtt{OPT}_{k_i}$.

  Then, according to the results of Theorem~\ref{thm:sdm:approx}, we have $ \mathtt{OPT} < 2\mu $ if $ S_{\mu} < k $ and $ \mathtt{OPT}_{k_i} < 2\mu $ if $ S_{\mu,i} < k_i $ for each $ i \in \{1,2\} $. Note that $ \mu^{\prime} $ is the largest $ \mu \in \mathcal{U}$ such that $|S_{\mu}|=k$, $|S_{\mu,1}|=k_1$, and $|S_{\mu,2}|=k_2$ after stream processing. For $\mu^{\prime\prime} = \frac{\mu^{\prime}}{1-\varepsilon} \in \mathcal{U}$, we have either $|S_{\mu^{\prime\prime}}| < k$ or $|S_{\mu^{\prime\prime},i}| < k_i$ for some $ i \in \{1,2\} $. Therefore, it holds that $ \mathtt{OPT}_f < 2\mu^{\prime\prime} \leq \frac{2}{1-\varepsilon} \cdot \mu^{\prime}$ and we conclude the proof.
\end{proof}

\begin{lemma}\label{lm:sfdm:12}
  For each $\mu \in \mathcal{U}^{\prime}$, $S_{\mu}$ must satisfy $div(S_{\mu}) \geq \frac{\mu}{2}$ and $|S_{\mu} \cap X_i| = k_i$ for both $i \in \{1,2\}$ after post-processing.
\end{lemma}
\begin{proof}
  The candidate $S_{\mu}$ before post-processing has exactly $k=k_1+k_2$ elements but may not contain $k_1$ elements from $X_1$ and $k_2$ elements from $X_2$. If $S_{\mu}$ has exactly $k_1$ elements from $X_1$ and $k_2$ elements from $X_2$ and thus the post-processing is skipped, we have $div(S_{\mu}) \geq \mu$ according to Theorem~\ref{thm:sdm:approx}. Otherwise, assuming that $ |S_{\mu} \cap X_1| = k_1^{\prime} < k_1 $, we will add $k_1-k_1^{\prime}$ elements from $S_{\mu,1}$ to $ S_{\mu} $ and remove $k_1-k_1^{\prime}$ elements from $S_{\mu} \cap X_2$ for ensuring the fairness constraint. In Line~\ref{ln:sfdm:post:3}, all the $k_1$ elements in $S_{\mu,1}$ can be selected for insertion. Since the minimum distance between any pair of elements in $S_{\mu,1}$ is at least $\mu$, we can find at most one element $x \in S_{\mu,1}$ such that $d(x,y) < \frac{\mu}{2} $ for each $y \in S_{\mu} \cap X_1$. This means that there are at least $k_1-k_1^{\prime}$ elements from $S_{\mu,1}$ whose distances to all the existing elements in $S_{\mu} \cap X_1$ are greater than $\frac{\mu}{2}$. Accordingly, after adding $k_1-k_1^{\prime}$ elements from $S_{\mu,1}$ to $S_{\mu}$ greedily, it still holds that $d(x,y) \geq \frac{\mu}{2}$ for any $x,y \in S_{\mu} \cap X_1$. In Line~\ref{ln:sfdm:post:2}, for each element $x \in S_{\mu} \cap X_2$, there is at most one (newly added) element $y \in S_{\mu} \cap X_1$ such that $d(x,y) < \frac{\mu}{2}$. Meanwhile, it is guaranteed that $y$ is the nearest neighbor of $x$ in $S_{\mu}$ in this case. So, in Line~\ref{ln:sfdm:post:2}, every $x \in S_{\mu} \cap X_2$ with $d(x,S_{\mu} \cap X_2) < \frac{\mu^{\prime}}{2}$ is removed, since there are at most $ k_1-k_1^{\prime}$ such elements and the one with the smallest $d(x,S_{\mu} \cap X_2)$ is removed at each step. Therefore, $S_{\mu}$ contains $k_1$ elements from $X_1$ and $k_2$ elements from $X_2$ and satisfies that $div(S_{\mu}) \geq \frac{\mu}{2}$ after post-processing.
\end{proof}

\begin{theorem}\label{thm:sfdm:1:approx}
  \textnormal{\textsf{SFDM1}} returns a $\frac{1-\varepsilon}{4}$-approximate solution for the fair diversity maximization problem.
\end{theorem}
\begin{proof}
  According to Lemmas~\ref{lm:sfdm:11} and~\ref{lm:sfdm:12}, $ div(S) \geq div(S_{\mu^{\prime}}) \geq \frac{\mu^{\prime}}{2} \geq \frac{1-\varepsilon}{4} \cdot \mathtt{OPT}_f$, where $\mu^{\prime}$ is the largest $\mu \in \mathcal{U}^{\prime}$.
\end{proof}

\noindent\textbf{Complexity Analysis:}
We analyze the time and space complexities of \textsf{SFDM1} in Theorem~\ref{thm:sfdm:1:complexity}.
\begin{theorem}\label{thm:sfdm:1:complexity}
  \textnormal{\textsf{SFDM1}} stores $O(\frac{k\log{\Delta}}{\varepsilon})$ elements in memory, takes $O(\frac{k\log{\Delta}}{\varepsilon})$ time per element for streaming processing, and spends $O(\frac{k^2\log{\Delta}}{\varepsilon})$ time for post-processing.
\end{theorem}
\begin{proof}
  \textnormal{\textsf{SFDM1}} keeps $3$ candidates for each $\mu \in \mathcal{U}$ and $O(k)$ elements in each candidate. Hence, the total number of stored elements is $O(\frac{k\log{\Delta}}{\varepsilon})$ since $|\mathcal{U}|=O(\frac{\log{\Delta}}{\varepsilon})$. The stream processing performs at most $O(\frac{k\log{\Delta}}{\varepsilon})$ distance computations per element. Finally, for each $\mu \in \mathcal{U}^{\prime}$ in the post-processing, at most $ k_i(k_i-k_i^{\prime}) $ distance computations are performed to select the elements in $S_{\mu,i}$ to be added to $S_{\mu}$. To find the elements to be removed, at most $ k(k_i-k_i^{\prime}) $ distance computations are needed. Therefore, the time complexity for post-processing is $O(\frac{k^2\log{\Delta}}{\varepsilon})$ since $|\mathcal{U}^{\prime}|=O(\frac{\log{\Delta}}{\varepsilon})$.
\end{proof}

\noindent\textbf{Comparison with Prior Art:}
The idea of finding an initial solution and balancing it for fairness in \textsf{SFDM1} has also been used for \textsf{FairSwap}~\cite{moumoulidou_et_al:LIPIcs.ICDT.2021.13}. However, \textsf{FairSwap} only works in the offline setting, which keeps the whole dataset in memory and needs random accesses over it for solution computation, whereas \textsf{SFDM1} works in the streaming setting, which scans the dataset in one pass and uses only the elements in the candidates for post-processing. Compared with \textsf{FairSwap}, \textsf{SFDM1} reduces the space complexity from $O(n)$ to $O(\frac{k\log{\Delta}}{\varepsilon})$ and the time complexity from $O(n k)$ to $O(\frac{k^2\log{\Delta}}{\varepsilon})$ at the expense of lowering the approximation ratio by a factor of $1-\varepsilon$.

\subsection{Streaming Algorithm for General \texorpdfstring{$m$}{m}}
\label{subsec:alg:stream:2}

\begin{algorithm}[t]
  \caption{\textsf{SFDM2}}\label{alg:sfdm:2}
  \Input{A stream $X = \bigcup_{i = 1}^{m} X_i$, a distance metric $d$, a parameter $\varepsilon \in (0,1)$, $m$ size constraints $k_1,\ldots,k_m \in \mathbb{Z}^+$ ($ k = \sum_{i=1}^{m} k_i $)}
  \Output{A set $S \subseteq X$ s.t.~$ |S \cap X_i|=k_i $, $\forall i \in [m]$}
  \tcc{Stream processing}
  $\mathcal{U} = \{\frac{d_{min}}{(1-\varepsilon)^j} \, : \, j \in \mathbb{Z}_{0}^{+} \wedge \frac{d_{min}}{(1-\varepsilon)^j} \in [d_{min}, d_{max}] \}$\;
  Initialize $ S_{\mu},S_{\mu, i} = \emptyset $ for every $ \mu \in \mathcal{U} $ and $i \in [m]$\;
  \ForEach{$x \in X$}{
    \ForEach{$\mu \in \mathcal{U}$ \textnormal{and} $i \in [m]$}{
      \If{$|S_{\mu}| < k \wedge d(x,S_{\mu}) \geq \mu$}{
        $S_{\mu} \gets S_{\mu} \cup \{x\}$\;
      }
      \If{$c(x)=i \wedge |S_{\mu, i}|<k \wedge d(x,S_{\mu, i}) \geq \mu$}{
        $S_{\mu, i} \gets S_{\mu, i} \cup \{x\} $\;
      }
    }
  }
  \tcc{Post-processing}
  $\mathcal{U}^{\prime} = \{ \mu \in \mathcal{U} : |S_{\mu}|=k \wedge |S_{\mu, i}|\geq k_i, \forall i \in [m] \}$\;
  \ForEach{$\mu \in \mathcal{U}^{\prime}$}{
    For each group $i \in [m]$, pick $\min(k_i,|S_{\mu} \cap X_i|)$ elements arbitrarily 
    from $S_{\mu}$ as $S^{\prime}_{\mu}$\;
    Let $ S_{all} = (\bigcup_{i=1}^{m} S_{\mu, i}) \cup S_{\mu} $ and $ l = |S_{all}| $\;
    Create $l$ clusters $ \mathcal{C}=\{C_1, \ldots, C_l\} $, each of which contains
    one element in $S_{all}$\;\label{ln:sfdm:2:s1}
    \While{there exist $C_a,C_b \in \mathcal{C} $ s.t.~$d(x,y) < \frac{\mu}{m+1}$ for some $x \in C_a$ and $y \in C_b$}{
      Merge $C_a,C_b$ into a new cluster $C^{\prime} = C_a \cup C_b$\;
      $\mathcal{C} \gets \mathcal{C} \setminus \{C_a,C_b\} \cup \{C^{\prime}\}$\;\label{ln:sfdm:2:s2}
    }
    Let $\mathcal{M}_1=(S_{all},\mathcal{I}_1)$ and $\mathcal{M}_2=(S_{all},\mathcal{I}_2)$
    be two matroids, where $S \in \mathcal{I}_1$ iff $|S \cap X_i| \leq k_i$, $\forall i \in [m]$
    and $S \in \mathcal{I}_2$ iff $ |S \cap C| \leq 1$, $\forall C \in \mathcal{C}$\;
    Run Algorithm~\ref{alg:matroid} to augment $S^{\prime}_{\mu}$ so that $S^{\prime}_{\mu}$ is a maximum cardinality set in $\mathcal{I}_1 \cap \mathcal{I}_2$\;
  }
  \Return{$ S \gets \argmax_{\mu \in \mathcal{U}^{\prime} : |S^{\prime}_{\mu}| = k} div(S^{\prime}_{\mu}) $}\;
\end{algorithm}

Now we introduce our streaming algorithm called \textsf{SFDM2} that can work with an arbitrary $m \geq 2$. The detailed procedure of \textsf{SFDM2} is presented in Algorithm~\ref{alg:sfdm:2}. Similar to \textsf{SFDM1}, it also has two phases: \emph{stream processing} and \emph{post-processing}. In stream processing, it utilizes Algorithm~\ref{alg:sdm} to keep a group-blind candidate $S_{\mu}$ and $m$ group-specific candidates $S_{\mu, 1}, \ldots, S_{\mu, m}$ for all the $m$ groups. The difference from \textsf{SFDM1} is that the size constraint of each group-specific candidate is $k$ instead of $k_i$ for each group $i$. Then, after processing all elements in $X$, it requires a post-processing scheme for ensuring the fairness of candidates as well. Nevertheless, the post-processing procedures are totally different from \textsf{SFDM1}, since the swap-based balancing strategy cannot guarantee the validity of the solution with any theoretical bound. Like \textsf{SFDM1}, the post-processing is performed on a subset $\mathcal{U}^{\prime}$ where $S_{\mu}$ has $k$ elements and $S_{\mu,i}$ has at least $k_i$ elements for each group $i$. For each $\mu \in \mathcal{U}^{\prime}$, it initializes with a subset $S^{\prime}_{\mu}$ of $S_{\mu}$. For an over-filled group $i$ -- i.e., $|S_{\mu} \cap X_i| > k_i$, $S^{\prime}_{\mu}$ contains $k_i$ arbitrary elements from $S_{\mu}$; For an under-filled or exactly filled group $i$ -- i.e., $|S_{\mu} \cap X_i| \leq k_i$, $S^{\prime}_{\mu}$ contains all $k_i^{\prime}=|S_{\mu} \cap X_i|$ elements from $S_{\mu}$. Next, new elements from under-filled groups should be added to $S^{\prime}_{\mu}$ so that $S^{\prime}_{\mu}$ is a fair solution. The method to find the elements to be added is to divide the set $S_{all}$ of elements in all candidates into a set $\mathcal{C}$ of clusters which guarantees that $d(x,y) \geq \frac{\mu}{m+1}$ for any $x \in C_a$ and $y \in C_b$, where $C_a$ and $C_b$ are two different clusters in $\mathcal{C}$. Then, $S^{\prime}_{\mu}$ is limited to contain at most one element from each cluster after new elements are added so that $div(S^{\prime}_{\mu}) \geq \frac{\mu}{m+1}$. Meanwhile, $S^{\prime}_{\mu}$ should still satisfy the fairness constraint. To meet both requirements, the problem of adding new elements to $S^{\prime}_{\mu}$ is formulated as an instance of \emph{matroid intersection}~\cite{DBLP:journals/siamcomp/Cunningham86,DBLP:conf/focs/ChakrabartyLS0W19,nguyen2019note} as will be discussed later. Finally, it returns $S^{\prime}_{\mu}$ containing $k$ elements with maximum diversity after post-processing as the final solution for FDM.

\begin{algorithm}[t]
  \caption{Matroid Intersection}\label{alg:matroid}
  \Input{Two matroids $\mathcal{M}_1=(V,\mathcal{I}_1)$, $\mathcal{M}_2=(V,\mathcal{I}_2)$, a distance metric $d$, an initial set $S_0 \subseteq V$}
  \Output{A maximum cardinality set $S \subseteq V$ in $\mathcal{I}_1 \cap \mathcal{I}_2$}
  Initialize $S \gets S_0$, $V_1 = \{x \in V \setminus S : S \cup \{x\} \in \mathcal{I}_1\}$, and
  $V_2 = \{x \in V \setminus S : S \cup \{x\} \in \mathcal{I}_2\}$\;
  \While{$V_1 \cap V_2 \neq \emptyset$\label{ln:matroid:s}}{
    $x^* \gets \argmax_{x \in V_1 \cap V_2} d(x,S)$ and $S \gets S \cup \{x^*\}$\;
    \ForEach{$x \in V_1$}{
      $V_1 \gets V_1 \setminus \{x\}$ if $S \cup \{x\} \notin \mathcal{I}_1$\;
    }
    \ForEach{$x \in V_2$}{
      $V_2 \gets V_2 \setminus \{x\}$ if $S \cup \{x\} \notin \mathcal{I}_2$\;\label{ln:matroid:t}
    }
  }
  Build an augmentation graph $G$ for $S$\;
  \While{there is a directed path from $a$ to $b$ in $G$}{
    Let $P^*$ be a shortest path from $a$ to $b$ in $G$\;
    \ForEach{$x \in P^* \setminus \{a, b\}$}{
      $S \gets S \cup \{x\}$ if $x \notin S$\;
      $S \gets S \setminus \{x\}$ if $x \in S$\;
    }
    Rebuild $G$ for the updated $S$\;
  }
  \Return{$S$}\;
\end{algorithm}

\begin{figure}[t]
  \centering
  \includegraphics[width=0.45\textwidth]{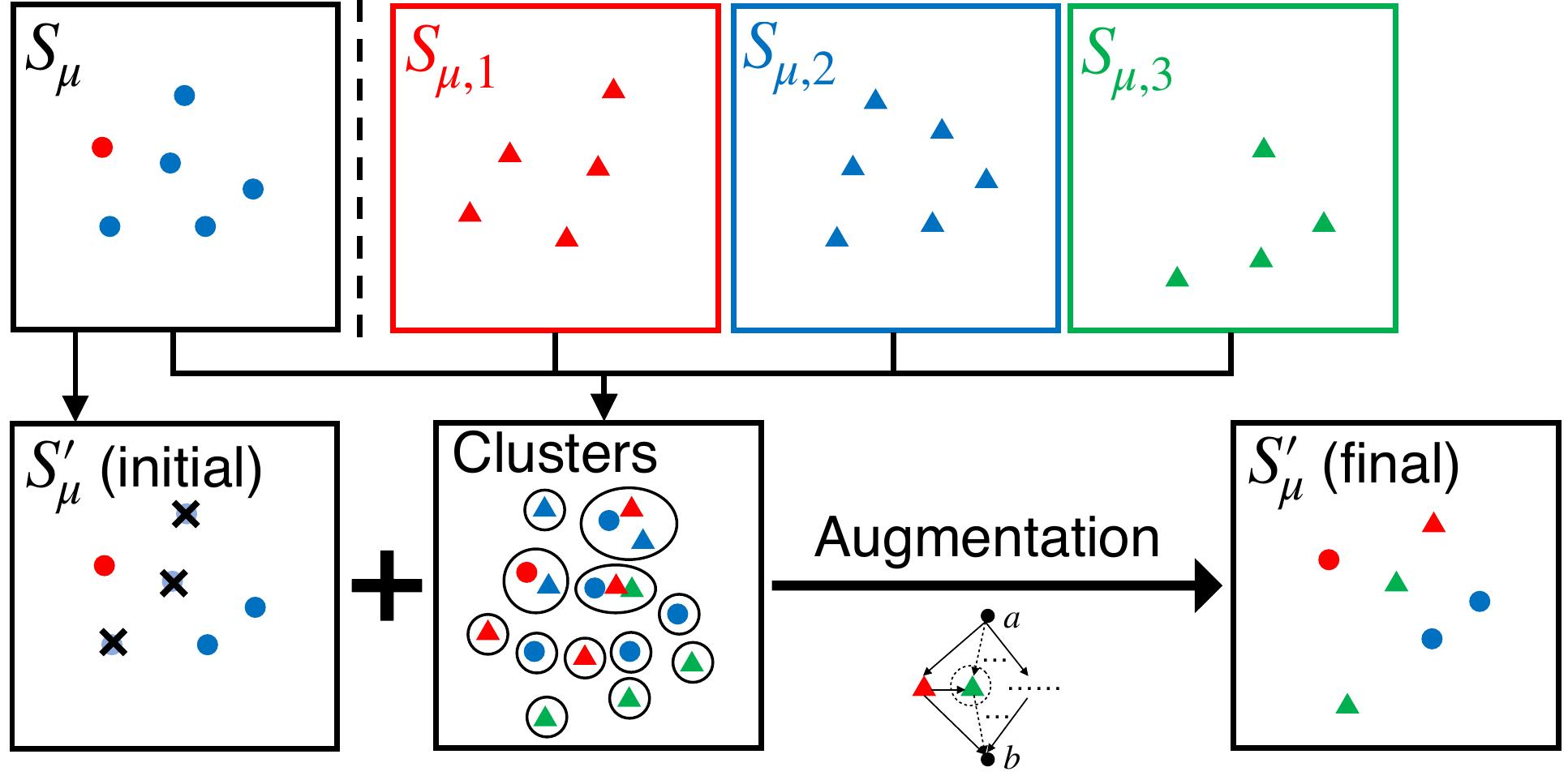}
  \caption{Illustration of post-processing in \textsf{SFDM2}. For each $\mu \in \mathcal{U}^{\prime}$, an initial $S^{\prime}_{\mu}$ is first extracted from $S_{\mu}$ by removing the elements from over-filled groups. Then, the elements in all candidates are divided into clusters. The final $S^{\prime}_{\mu}$ is augmented from the initial solution by adding new elements from under-filled groups based on matroid intersection.}
  \label{fig:alg2}
\end{figure}

\noindent\textbf{Matroid Intersection:}
Next, we describe how to use \emph{matroid intersection} for solution augmentation in \textsf{SFDM2}. We define the first rank-$k$ matroid $\mathcal{M}_1 = (V, \mathcal{I}_1)$ based on the fairness constraint, where the ground set $V$ is $S_{all}$ and $S \in \mathcal{I}_1$ iff $|S \cap X_i| \leq k_i$, $\forall i \in [m]$. Intuitively, a set $S$ is fair iff it is a maximal independent set in $\mathcal{I}_1$. Moreover, we define the second rank-$l$ ($l=|\mathcal{C}|$) matroid $\mathcal{M}_2=(V=S_{all},\mathcal{I}_2)$ on the set $\mathcal{C}$ of clusters, where $S \in \mathcal{I}_2$ iff $|S \cap C| \leq 1$, $\forall C \in \mathcal{C}$. Accordingly, the problem of adding new elements to $S^{\prime}_{\mu}$ for ensuring fairness can be seen as a \emph{matroid intersection} problem, which aims to find a maximum cardinality set $S \in \mathcal{I}_1 \cap \mathcal{I}_2$ for $\mathcal{M}_1=(S_{all},\mathcal{I}_1)$ and $\mathcal{M}_2=(S_{all},\mathcal{I}_2)$. A common solution for matroid intersection is the Cunningham's algorithm~\cite{DBLP:journals/siamcomp/Cunningham86} based on the \emph{augmentation graph} in Definition~\ref{def:aug:graph}.

\begin{definition}[Augmentation Graph~\cite{DBLP:journals/siamcomp/Cunningham86}]
\label{def:aug:graph}
  Given two matroids $\mathcal{M}_1=(V,\mathcal{I}_1)$ and $\mathcal{M}_2=(V,\mathcal{I}_2)$, a set $S \subset V$ such that $S \in \mathcal{I}_1 \cap \mathcal{I}_2$, and two sets $V_1=\{x \in V \setminus S : S \cup \{x\} \in \mathcal{I}_1\}$ and $V_2=\{x \in V \setminus S : S \cup \{x\} \in \mathcal{I}_2\}$, an augmentation graph is a digraph $G=(V \cup \{a, b\},E)$ where $a, b \notin V$. There is an edge $(a, x) \in E$ for each $x \in V_1$. There is an edge $(x, b) \in E$ for each $x \in V_2$. There is an edge $(y, x) \in E$ for each $x \in V \setminus S$, $y \in S$, such that $ S \cup \{x\} \notin \mathcal{I}_1 $ and $ S \cup \{x\} \setminus \{y\} \in \mathcal{I}_1$. There is an edge $(x, y) \in E$ for each $x \in V \setminus S$, $y \in S$, such that $ S \cup \{x\} \notin \mathcal{I}_2 $ and $ S \cup \{x\} \setminus \{y\} \in \mathcal{I}_2$.
\end{definition}

The Cunningham's algorithm~\cite{DBLP:journals/siamcomp/Cunningham86} is initialized with $S = \emptyset$ (or any $S \in \mathcal{I}_1 \cap \mathcal{I}_2$). At each step, it builds an augmentation graph $G$ for $\mathcal{M}_1$, $\mathcal{M}_2$, and $S$. If there is no directed path from $a$ to $b$ in $G$, then $S$ is returned as a maximum cardinality set. Otherwise, it finds the shortest path $P^*$ from $a$ to $b$ in $G$, and augments $S$ according to $P^*$: For each $x \in V \setminus S$ except $a$ and $b$, add $x$ to $S$; For each $x \in S$, remove $x$ from $S$.

We adapt the Cunningham's algorithm~\cite{DBLP:journals/siamcomp/Cunningham86} for our problem as shown in Algorithm~\ref{alg:matroid}. Our algorithm is initialized with $S^{\prime}_{\mu}$ instead of $\emptyset$. In addition, to reduce the cost of building $G$ and maximize the diversity, it first add the elements in $V_1 \cap V_2$ greedily to $S^{\prime}_{\mu}$ until $V_1 \cap V_2 = \emptyset$. This is because \emph{there exists a shortest path $P^*=\langle a, x, b \rangle$ in $G$ for any $x \in V_1 \cap V_2$}, which is easy to verify from Definition~\ref{def:aug:graph}. Finally, if $|S| < k$ after the above procedures, the standard Cunningham's algorithm will be used to augment $S$ for ensuring its maximality.

\vspace{1mm}
\noindent\textbf{Approximation Ratio:}
Next, we prove that \textsf{SFDM2} achieves an approximation ratio of $\frac{1-\varepsilon}{3m+2}$ for FDM. For the proof, we first show that the set $\mathcal{C}$ of clusters has several important properties (Lemma~\ref{lm:sfdm:21}). Then, we prove that Algorithm~\ref{alg:matroid} can return a fair solution for a specific $\mu$ based on the properties of $\mathcal{C}$ (Lemma~\ref{lm:sfdm:22}).

\begin{lemma}\label{lm:sfdm:21}
  The set $\mathcal{C}$ of clusters has the following properties: (i) for any element $x \in C_a$ and $y \in C_b$ ($a \neq b$), $d(x,y) \geq \frac{\mu}{m+1}$; (ii) each cluster $C$ contains at most one element from $S_{\mu}$ and $S_{\mu, i}$ for any $i \in [m]$; (iii) for any two elements $x,y \in C$, $d(x,y) < \frac{m}{m+1} \cdot \mu$.
\end{lemma}
\begin{proof}
  First of all, Property (i) holds from Lines~\ref{ln:sfdm:2:s1}--\ref{ln:sfdm:2:s2} of Algorithm~\ref{alg:sfdm:2}. Then, we prove Property (ii) by contradiction. Let us construct an undirected graph $G=(V,E)$ for a cluster $C \in \mathcal{C}$, where $V$ is the set of elements in $C$ and there exists an edge $(x,y) \in E$ iff $d(x,y) < \frac{\mu}{m}$. Based on Algorithm~\ref{alg:sfdm:2}, for any $x \in C$, there must exist some $y \in C$ ($x \neq y$) such that $d(x,y) < \frac{\mu}{m}$. Therefore, $G$ is a connected graph. Suppose that $C$ can contain more than one element from $S_{\mu}$ or $S_{\mu, i}$ for some $i \in [m]$. Let $ P_{x,y} = (x,\ldots,y) $ be the shortest path of $G$ between $x$ and $y$ where $x$ and $y$ are both from $S_{\mu}$ or $S_{\mu, i}$. Next, we show that the length of $P_{x,y}$ is at most $m+1$. If the length of $P_{x,y}$ is longer than $m+1$, there will be a sub-path $P_{x^{\prime},y^{\prime}}$ of $P_{x,y}$ where $x^\prime$ and $y^\prime$ are both from $S_{\mu}$ or $S_{\mu, i}$ and this violates the fact that $P_{x,y}$ is the shortest. Since the length of $P_{x,y}$ is at most $m+1$, we have $d(x,y) < (m+1) \cdot \frac{\mu}{m+1} = \mu$, which contradicts with the fact that $d(x,y) \geq \mu$, as they are both from $S_{\mu}$ or $S_{\mu, i}$. Finally, Property (iii) is a natural extension of Property (ii): Since each cluster $C$ contains at most one element from $S_{\mu}$ and $S_{\mu, i}$ for any $i \in [m]$, $C$ has at most $m+1$ elements. So, for any two elements $x,y \in C$, the length of the path between them is at most $m$ in $G$ and $d(x,y) < m \cdot \frac{\mu}{m+1} = \frac{m}{m+1} \cdot \mu$.
\end{proof}

\begin{lemma}\label{lm:sfdm:22}
  If $\mathtt{OPT}_f \geq \frac{3m+2}{m+1}\cdot\mu$, then Algorithm~\ref{alg:matroid} returns a size-$k$ subset $S^{\prime}_{\mu}$ such that $S^{\prime}_{\mu} \in \mathcal{I}_1 \cap \mathcal{I}_2$ and $div(S^{\prime}_{\mu}) \geq \frac{\mu}{m+1}$.
\end{lemma}
\begin{proof}
  First of all, the initial $S^{\prime}_{\mu}$ is a subset of $S_{\mu}$. According to Property (ii) of Lemma~\ref{lm:sfdm:21}, all elements of $S^{\prime}_{\mu}$ are in different clusters of $\mathcal{C}$ and thus $S^{\prime}_{\mu} \in \mathcal{I}_1 \cap \mathcal{I}_2$. The analysis of~\cite{DBLP:journals/siamcomp/Cunningham86} guarantees that Algorithm~\ref{alg:matroid} can find a size-$k$ set in $\mathcal{I}_1 \cap \mathcal{I}_2$ as long as it exists. Next, we will show such a set exists when $\mathtt{OPT}_f \geq \frac{3m+2}{m+1}\cdot\mu$. To verify this, we need to identify $k_i$ clusters of $\mathcal{C}$ that contain at least one element from $X_i$ for each $i\in[m]$ and show that all $k=\sum_{i=1}^{m} k_i$ clusters are distinct. Here, we consider two cases for each group $i\in[m]$.
  \begin{itemize}[leftmargin=*]
    \item \textbf{Case 1:} For each $i \in [m]$ such that $k_i \leq |S_{\mu, i}| < k$, we have $d(x,S_{\mu, i}) < \mu$ for each $x \in X_i$. Given the optimal solution $S^*_f$, we define a function $f$ that maps each $x^* \in S^*_f$ to its nearest neighbor in $S_{\mu, i}$. For two elements $x^*_a, x^*_b \in S^*_f$ in these groups, we have $d(x^*_a, f(x^*_a)) < \mu$, $d(x^*_b, f(x^*_b)) < \mu$, and $d(x^*_a, x^*_b) \geq \mathtt{OPT}_f = div(S^*_f)$. Therefore, $ d(f(x^*_a),f(x^*_b)) > \mathtt{OPT}_f - 2\mu $. Since $\mathtt{OPT}_f \geq \frac{3m+2}{m+1}\cdot\mu$, $ d(f(x^*_a),f(x^*_b)) > \frac{3m+2}{m+1}\cdot\mu - 2\mu = \frac{m}{m+1}\cdot\mu$. According to Property (iii) of Lemma~\ref{lm:sfdm:21}, it is guaranteed that $f(x^*_a)$ and $f(x^*_b)$ are in different clusters. By identifying all the clusters that contains $f(x^*)$ for all $x^* \in S^*_f$, we have found $k_i$ clusters for each group $i \in [m]$ such that $k_i \leq |S_{\mu,i}| < k$. And all the clusters found are guaranteed to be distinct.
    \item \textbf{Case 2:} For all $i \in [m]$ such that $|S_{\mu,i}| = k$, we are able to find $k$ clusters that contain one element from $S_{\mu, i}$ based on Property (ii) of Lemma~\ref{lm:sfdm:21}. For such a group $i$, even though $k-k_i$ clusters have been identified for all other groups, there are still at least $k_i$ clusters available for selection. Therefore, we can always find $k_i$ clusters that are distinct from all the clusters identified by any other group for such a group $X_i$.
  \end{itemize}
  Considering both cases, we have proven the existence of a size-$k$ set in $\mathcal{I}_1 \cap \mathcal{I}_2$. Finally, for any set $S \in \mathcal{I}_2$, we have $div(S) \geq \frac{\mu}{m+1}$ according to Property (i) of Lemma~\ref{lm:sfdm:21}.
\end{proof}

\begin{theorem}\label{thm:sfdm:2:approx}
  \textnormal{\textsf{SFDM2}} achieves a $\frac{1-\varepsilon}{3m+2}$-approximation for the fair diversity maximization problem.
\end{theorem}
\begin{proof}
  Let $\mu^{-}$ be the smallest $\mu$ not in $\mathcal{U}^{\prime}$. It holds that $\mu^{-} \geq \frac{\mathtt{OPT}_f}{2}$ (see Lemma~\ref{lm:sfdm:11}). Thus, there is some $\mu < \mu^{-}$ in $\mathcal{U}^{\prime}$ such that $\mu \in [\frac{(m+1)(1-\varepsilon)}{3m+2}\cdot\mathtt{OPT}_f, \frac{m+1}{3m+2}\cdot\mathtt{OPT}_f]$, as $\frac{m+1}{3m+2} < \frac{1}{2}$ for any $m \in \mathbb{Z}^{+}$. Therefore, \textsf{SFDM2} provides a fair solution $S$ such that $div(S) \geq div(S^{\prime}_{\mu}) \geq \frac{\mu}{m+1} \geq \frac{1-\varepsilon}{3m+2}\cdot\mathtt{OPT}_f$.
\end{proof}

\noindent\textbf{Complexity Analysis:}
We analyze the time and space complexities of \textsf{SFDM2} in Theorem~\ref{thm:sfdm:2:complexity}.
\begin{theorem}\label{thm:sfdm:2:complexity}
  \textnormal{\textsf{SFDM2}} stores $O(\frac{km\log{\Delta}}{\varepsilon})$ elements, takes $ O(\frac{k\log{\Delta}}{\varepsilon}) $ time per element for streaming processing, and spends $ O\big(\frac{k^2m\log{\Delta}}{\varepsilon} \cdot (m+\log^2{k})\big) $ time for post-processing.
\end{theorem}
\begin{proof}
  \textsf{SFDM2} keeps $m+1$ candidates for each $\mu \in \mathcal{U}$ and $O(k)$ elements in each candidate. So, the total number of elements stored by \textsf{SFDM2} is $O(\frac{km\log{\Delta}}{\varepsilon})$. Only $2$ candidates are checked in streaming processing for each element and thus $O(\frac{k\log{\Delta}}{\varepsilon})$ distance computations are needed. In the post-processing of each $\mu$, we need $O(k)$ time to get the initial solution, $O(k^2m^2)$ time to cluster $S_{all}$, and $O(k^2m)$ time to augment the candidate using Lines~\ref{ln:matroid:s}--\ref{ln:matroid:t} of Algorithm~\ref{alg:matroid}. The time complexity of the Cunningham's algorithm is $O(k^2m \log^2{k})$ according to~\cite{nguyen2019note, DBLP:conf/focs/ChakrabartyLS0W19}. To sum up, the overall time complexity of post-processing is $ O\big(\frac{k^2m\log{\Delta}}{\varepsilon} \cdot (m+\log^2{k})\big) $.
\end{proof}

\noindent\textbf{Comparison with Prior Art:}
Finding a fair solution based on \emph{matroid intersection} has been used by existing methods for fair $k$-center~\cite{DBLP:conf/icml/JonesNN20, DBLP:journals/algorithmica/ChenLLW16, DBLP:conf/icml/ChiplunkarKR20} and fair diversity maximization~\cite{moumoulidou_et_al:LIPIcs.ICDT.2021.13}. \textsf{SFDM2} adopts a similar method to \textsf{FairFlow}~\cite{moumoulidou_et_al:LIPIcs.ICDT.2021.13} to construct the clusters and matroids. But \textsf{FairFlow} solves \emph{matroid intersection} as a max-flow problem on a digraph. Its solution is of poor quality in practice, particularly so when $m$ is large. Thus, \textsf{SFDM2} uses a different method for \emph{matroid intersection} based on the Cunningham's algorithm, which initializes with a partial solution instead of $\emptyset$ for higher efficiency and adds elements greedily like \textsf{GMM}~\cite{DBLP:journals/tcs/Gonzalez85} for higher diversity. Hence, \textsf{SFDM2} has significantly higher solution quality than \textsf{FairFlow} though its approximation ratio is lower.

\section{Experiments}
\label{sec:exp}

In this section, we evaluate the performance of our proposed algorithms on several real-world and synthetic datasets. We first introduce our experimental setup in Section~\ref{subsec:setup}. Then, the experimental results are presented in Section~\ref{subsec:results}.

\begin{table}[t]
\centering
\caption{Statistics of datasets in the experiments}
\label{tbl:stat}
\begin{tabular}{|c|c|c|c|c|}
\hline
\textbf{dataset}   & $n$            & $m$          & \# \textbf{features} & \textbf{distance metric} \\ \hline
\textbf{Adult}     & $48,842$       & $2$/$5$/$10$ & $6$                  & Euclidean                \\ \hline
\textbf{CelebA}    & $202,599$      & $2$/$4$      & $41$                 & Manhattan                \\ \hline
\textbf{Census}    & $2,426,116$    & $2$/$7$/$14$ & $25$                 & Manhattan                \\ \hline
\textbf{Lyrics}    & $122,448$      & $15$         & $50$                 & Angular                  \\ \hline
\textbf{Synthetic} & $10^3$--$10^7$ & $2$--$20$    & $2$                  & Euclidean                \\ \hline
\end{tabular}
\end{table}

\begin{table*}[t]
\centering
\caption{Overview of the performance of different algorithms ($k=20$)}
\label{tbl:exp:overview}
\begin{tabular}{|c|c|c|c|c|c|c|c|c|c|c|c|c|c|}
\hline
\multirow{2}{*}{\textbf{Dataset}} &
\multirow{2}{*}{\textbf{Group}} &
\multirow{2}{*}{$m$} &
\textsf{GMM} &
\multicolumn{2}{c|}{\textsf{FairSwap}} &
\multicolumn{2}{c|}{\textsf{FairFlow}} &
\multicolumn{3}{c|}{\textsf{SFDM1}} &
\multicolumn{3}{c|}{\textsf{SFDM2}} \\ \cline{4-14} 
& & & diversity & diversity & time(s) & diversity & time(s) & diversity & time(s) & \#elem & diversity & time(s) & \#elem \\ \hline
\multirow{3}{*}{Adult} & Sex & 2 & \multirow{3}{*}{5.0226} & 4.1485 & 9.583 & 3.1190 & 7.316 & 3.9427 & \textbf{0.040} & 90.2 & \textbf{4.1710} & 0.133 & 120.4 \\ \cline{2-3} \cline{5-14} 
& Race & 5 & & - & - & 1.3702 & 7.951 & - & - & - & \textbf{3.1373} & \textbf{1.435} & 312.3 \\ \cline{2-3} \cline{5-14} 
& Sex+Race & 10 & & - & - & 1.0049 & 8.732 & - & - & - & \textbf{2.9182} & \textbf{4.454} & 620.6 \\ \hline
\multirow{3}{*}{CelebA} & Sex & 2 & \multirow{3}{*}{13.0} & \textbf{11.4} & 34.892 & 8.4 & 23.257 & 9.8 & \textbf{0.018} & 87.2 & 10.9 & 0.039 & 122.3 \\ \cline{2-3} \cline{5-14} 
& Age & 2 & & \textbf{11.4} & 36.606 & 7.2 & 26.660 & 10.4 & \textbf{0.025} & 94.6 & 10.8 & 0.0672 & 128.0 \\ \cline{2-3} \cline{5-14} 
& Sex+Age & 4 & & - & - & 6.3 & 23.950 & - & - & - & \textbf{10.4} & \textbf{0.107} & 193.1 \\ \hline
\multirow{3}{*}{Census} & Sex & 2 & \multirow{3}{*}{35.0} & 27.0 & 355.315 & 17.5 & 246.518 & 27.0 & \textbf{0.032} & 121.5 & \textbf{31.0} & 0.089 & 163.0 \\ \cline{2-3} \cline{5-14}
& Age & 7 & & - & - & 8.5 & 297.923 & - & - & - & \textbf{21.0} & \textbf{0.797} & 676.0 \\ \cline{2-3} \cline{5-14} 
& Sex+Age & 14 & & - & - & 5.0 & 415.363 & - & - & - & \textbf{19.0} & \textbf{4.193} & 1276.0 \\ \hline
Lyrics & Genre & 15 & 1.5476 & - & - & 0.2228 & 18.239 & - & - & - & \textbf{1.4528} & \textbf{3.224} & 675.4 \\ \hline
\end{tabular}
\end{table*}

\subsection{Experimental Setup}
\label{subsec:setup}

\noindent\textbf{Datasets:}
We perform our experiments on four publicly available real-world datasets as follows:
\begin{itemize}[leftmargin=*]
  \item \textbf{Adult}\footnote{\url{https://archive.ics.uci.edu/ml/datasets/adult}} is a collection of $48,842$ records extracted from the 1994 US Census database. We select $6$ numeric attributes as features and normalize each of them to have zero mean and unit standard deviation. The Euclidean distance is used as the distance metric. The groups are generated from two demographic attributes: \emph{sex} and \emph{race}. By using them individually and in combination, there are $2$ (\emph{sex}), $5$ (\emph{race}), and $10$ (\emph{sex+race}) groups, respectively.
  \item \textbf{CelebA}\footnote{\url{https://mmlab.ie.cuhk.edu.hk/projects/CelebA.html}} is a set of $202,599$ images of human faces. We use $41$ pre-trained class labels as features and the Manhattan distance as the distance metric. We generate $2$ groups from \emph{sex} \{`female', `male'\}, $2$ groups from \emph{age} \{`young', `not young'\}, and $4$ groups from both of them, respectively.
  \item \textbf{Census}\footnote{\url{https://archive.ics.uci.edu/ml/datasets/US+Census+Data+(1990)}} is a set of $2,426,116$ records obtained from the 1990 US Census data. We take $25$ (normalized) numeric attributes as features and use the Manhattan distance as the distance metric. We generate $2$, $7$, and $14$ groups from \emph{sex}, \emph{age}, and both of them, respectively.
  \item \textbf{Lyrics}\footnote{\url{http://millionsongdataset.com/musixmatch}} is a collection of $122,448$ documents, each of which is the lyrics of a song. We train a topic model with $50$ topics using LDA~\cite{DBLP:journals/jmlr/BleiNJ03} implemented in Gensim\footnote{\url{https://radimrehurek.com/gensim}}. Each document is represented as a $50$-dimensional feature vector and the angular distance is used as the distance metric. We generate $15$ groups based on the primary genres of songs.
\end{itemize}
We generate different synthetic datasets with varying $n$ and $m$ for scalability tests. In each synthetic dataset, we generate ten $2$-dimensional Gaussian isotropic blobs with random centers in $[-10,10]^2$ and identity covariance matrices. We assign points to groups uniformly at random. The Euclidean distance is used as the distance metric. The total number $n$ of points varies from $10^3$ to $10^7$ with fixed $m=2$ or $10$. And the number $m$ of groups varies from $2$ to $20$ with fixed $n=10^5$. The statistics of all datasets are summarized in Table~\ref{tbl:stat}.

\begin{figure*}[!t]
  \centering
  \subfigure[Adult (Sex, $m=2$)]{
    \includegraphics[width=0.155\textwidth]{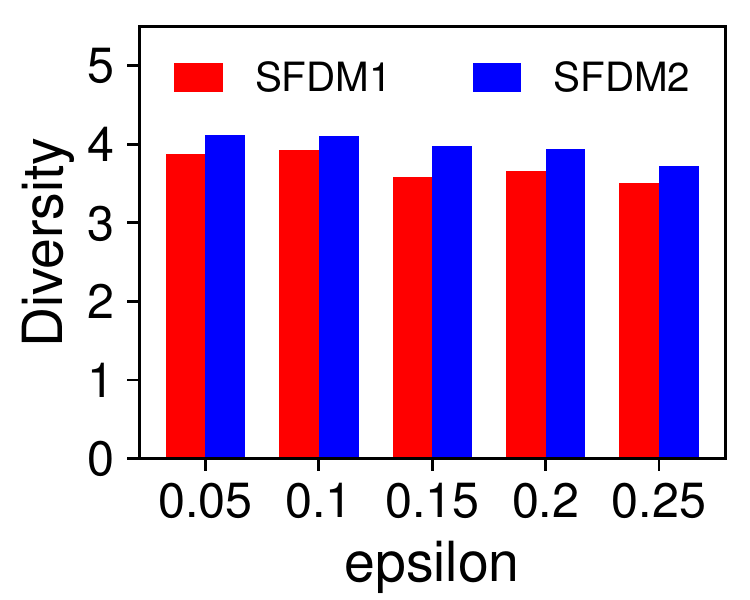}
    \includegraphics[width=0.155\textwidth]{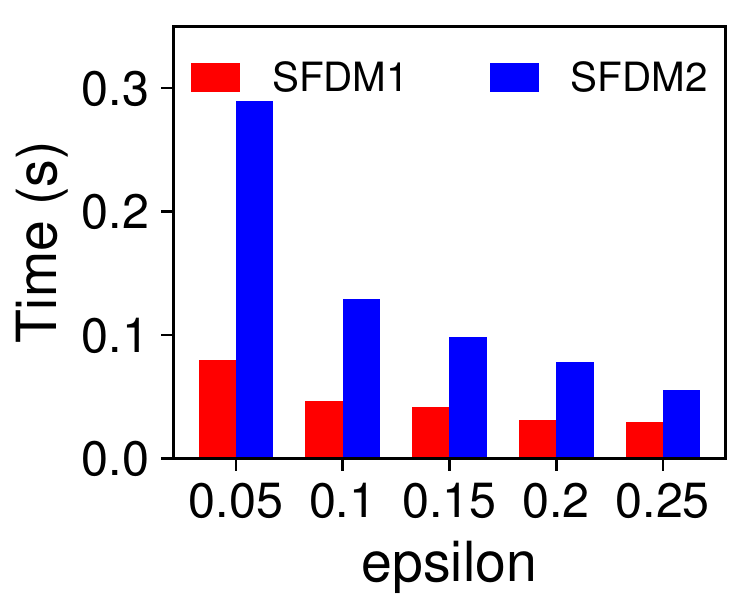}
    \includegraphics[width=0.155\textwidth]{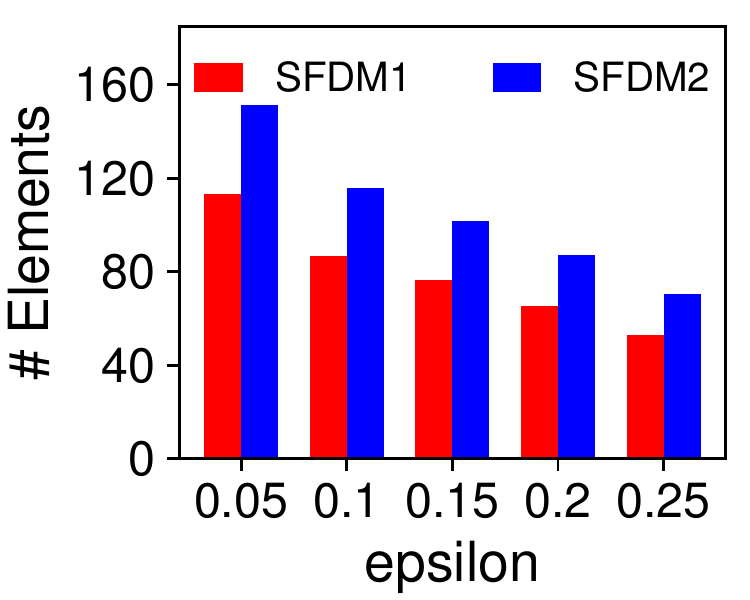}
  }
  \hfill
  \subfigure[CelebA (Sex, $m=2$)]{
    \includegraphics[width=0.155\textwidth]{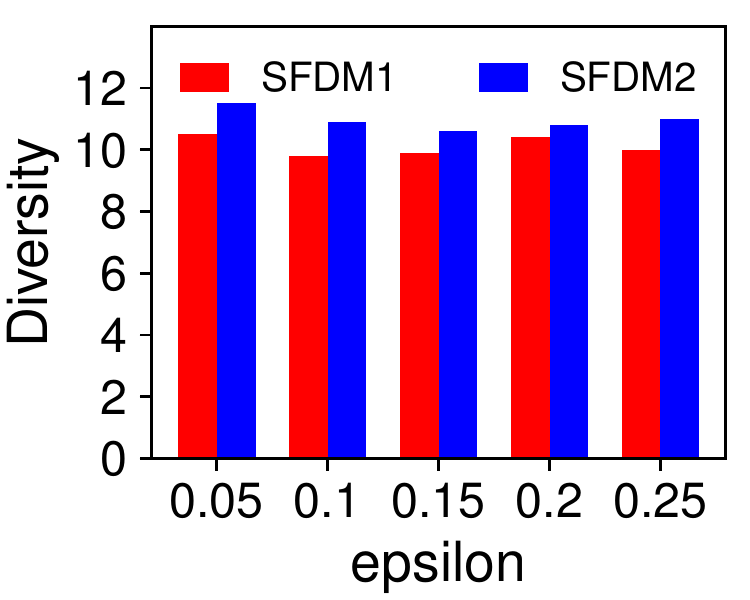}
    \includegraphics[width=0.155\textwidth]{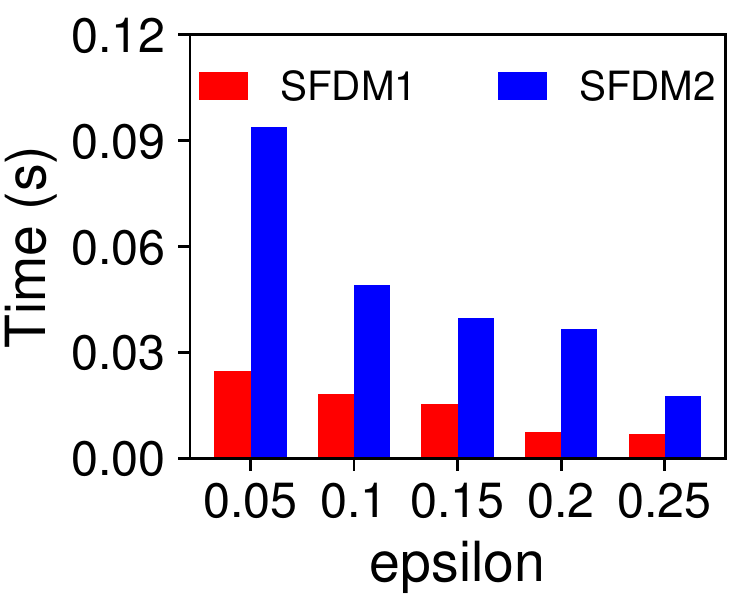}
    \includegraphics[width=0.155\textwidth]{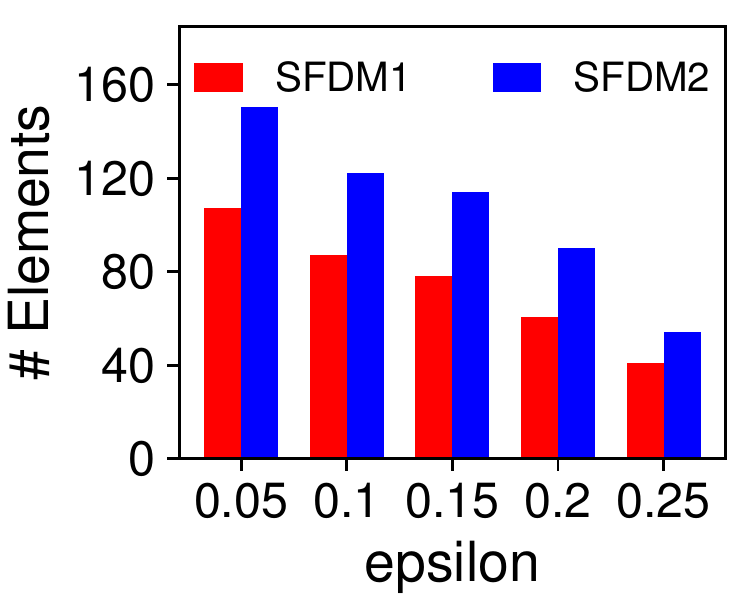}
  }
  \subfigure[Census (Sex, $m=2$)]{
    \includegraphics[width=0.155\textwidth]{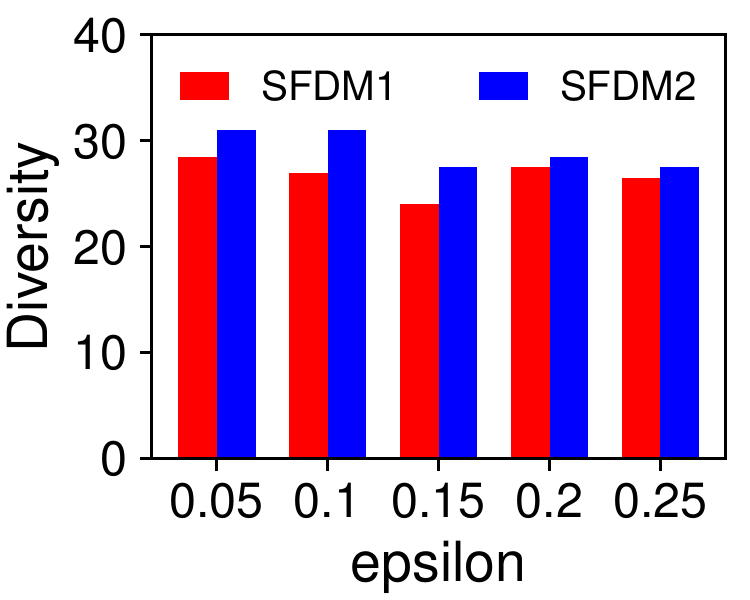}
    \includegraphics[width=0.155\textwidth]{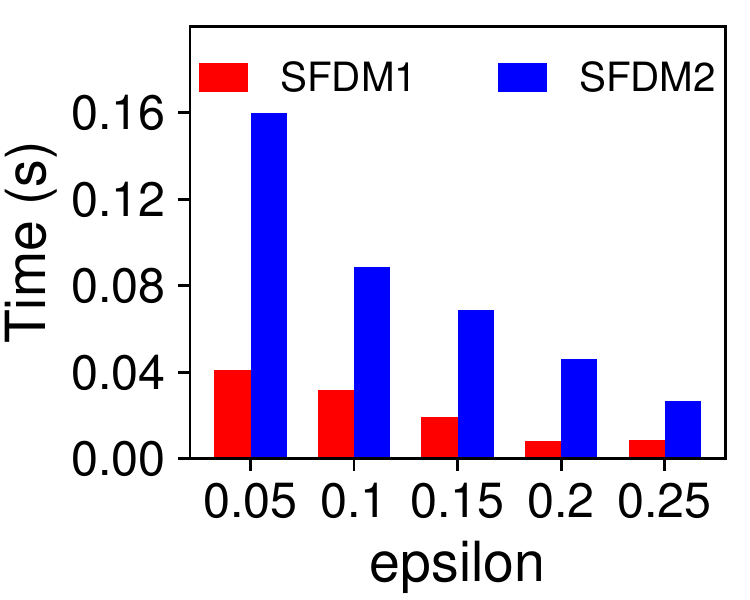}
    \includegraphics[width=0.155\textwidth]{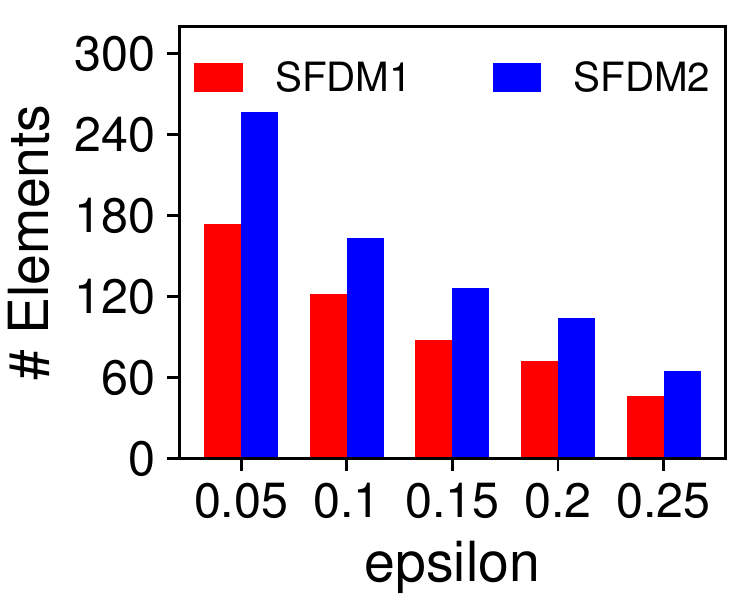}
  }
  \hfill
  \subfigure[Lyrics (Genre, $m=15$)]{
    \includegraphics[width=0.155\textwidth]{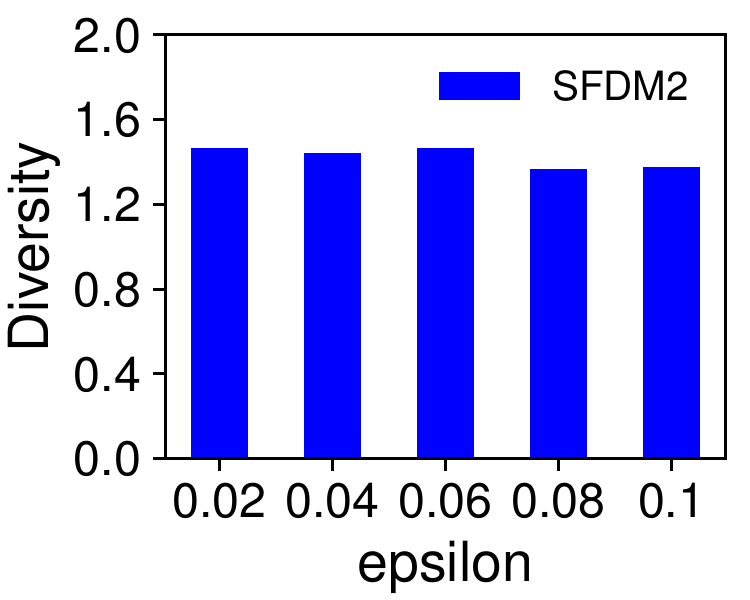}
    \includegraphics[width=0.155\textwidth]{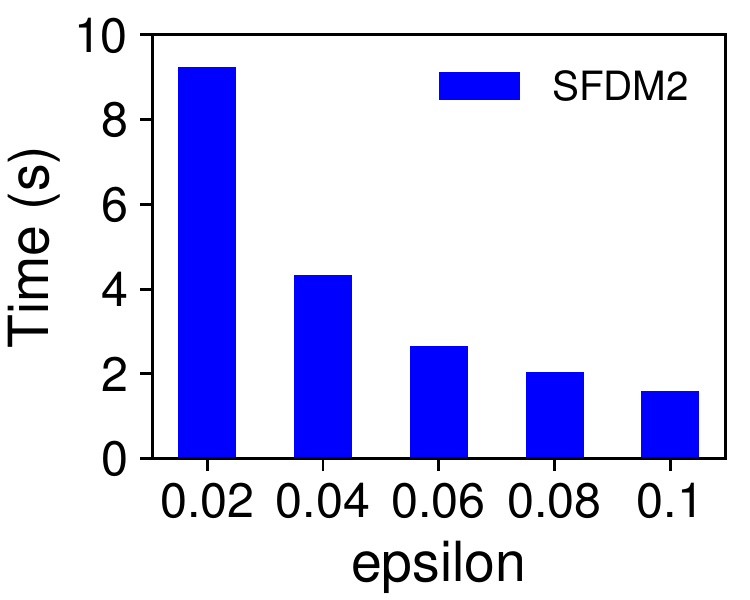}
    \includegraphics[width=0.155\textwidth]{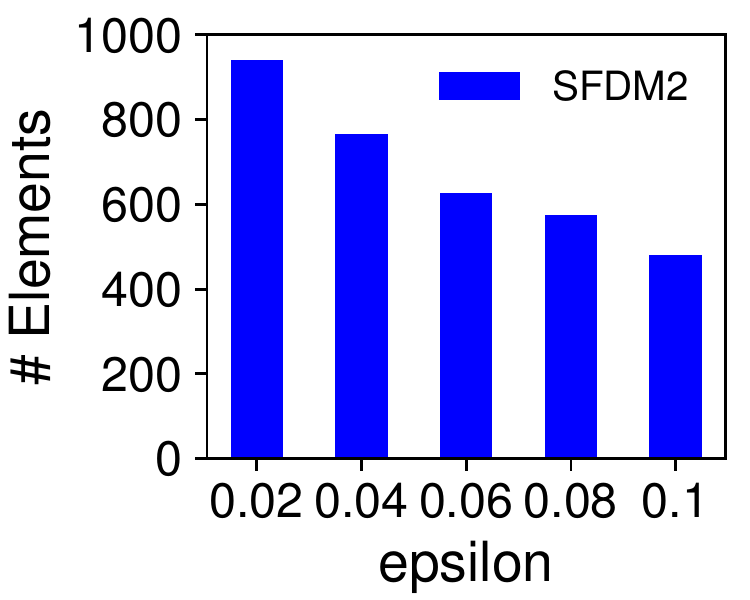}
  }
  \caption{Performance of \textsf{SFDM1} and \textsf{SFDM2} with varying parameter $\varepsilon$ ($k=20$).}
  \label{fig:exp:eps}
\end{figure*}

\vspace{1mm}
\noindent\textbf{Algorithms:}
We compare our streaming algorithms -- i.e., \textsf{SFDM1} and \textsf{SFDM2}, with three offline algorithms for FDM in~\cite{moumoulidou_et_al:LIPIcs.ICDT.2021.13}: the $\frac{1}{3m-1}$-approximation \textsf{FairFlow} algorithm for an arbitrary $m$, the $\frac{1}{5}$-approximation \textsf{FairGMM} algorithm for small $k$ and $m$, and the $\frac{1}{4}$-approximation \textsf{FairSwap} algorithm for $m=2$. Since no implementation for the algorithms in~\cite{moumoulidou_et_al:LIPIcs.ICDT.2021.13} is available, they are implemented by ourselves following the description of the original paper. We implement all the algorithms in Python 3.8. Our code is published on GitHub\footnote{\url{https://github.com/yhwang1990/code-FDM}}. All the experiments are run on a server with an Intel Broadwell 2.40GHz CPU and 29GB memory running Ubuntu 16.04.

For each experiment, \textsf{SFDM1} and \textsf{SFDM2} are invoked with parameter $\varepsilon = 0.1$ ($\varepsilon = 0.05$ for \emph{Lyrics}) by default. For a given size constraint $k$, the group-specific size constraint $k_i$ for each group $i \in [m]$ is set based on \emph{equal representation}, which has been widely used in the literature~\cite{DBLP:conf/www/0001FM21,DBLP:conf/icml/KleindessnerAM19,DBLP:conf/icml/ChiplunkarKR20,DBLP:conf/icml/JonesNN20}: If $k$ is divisible by $m$, $k_i = \frac{k}{m}$ for each $i \in [m]$; If $k$ is not divisible by $m$, $k_i = \lceil \frac{k}{m} \rceil$ for some groups or $k_i = \lfloor \frac{k}{m} \rfloor$ for the others with $\sum_{i=1}^{m}=k$. We also compare the performance of different algorithms for \emph{proportional representation}~\cite{DBLP:conf/icml/CelisKS0KV18,NEURIPS2020_9d752cb0,DBLP:conf/www/0001FM21}, another popular notion of fairness that requires the proportion of elements from each group in the solution generally preserves that in the original dataset.

\vspace{1mm}
\noindent\textbf{Performance Measures:}
The performance of each algorithm is evaluated in terms of \emph{efficiency}, \emph{quality}, and \emph{space usage}. The efficiency is measured as \emph{average update time} -- i.e., the average wall-clock time used to compute a solution for each arrival element in the stream. The quality is measured by the value of the \emph{diversity} function of the solution returned by an algorithm. Since computing the optimal diversity $\mathtt{OPT}_f$ of FDM is infeasible, we run the \textsf{GMM} algorithm~\cite{DBLP:journals/tcs/Gonzalez85} for unconstrained diversity maximization to estimate an upper bound of $\mathtt{OPT}_f$ for comparison. The space usage is measured by the number of distinct elements stored by each algorithm. Only the space usages of \textsf{SFDM1} and \textsf{SFDM2} are reported because the offline algorithms keep all elements in memory for random access and thus their space usages are always equal to the dataset size. We run each experiment $10$ times with different permutations of the same dataset and report the average of each measure over $10$ runs for evaluation.

\begin{figure*}
  \centering
  \includegraphics[width=0.6\textwidth]{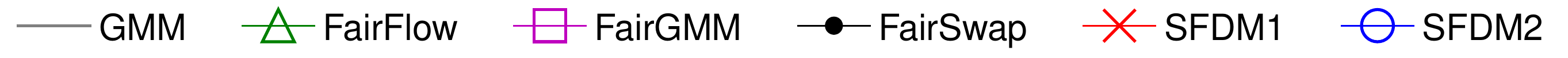}
  \\
  \subfigure[Adult (Sex, $m=2$)]{
    \includegraphics[width=0.2\textwidth]{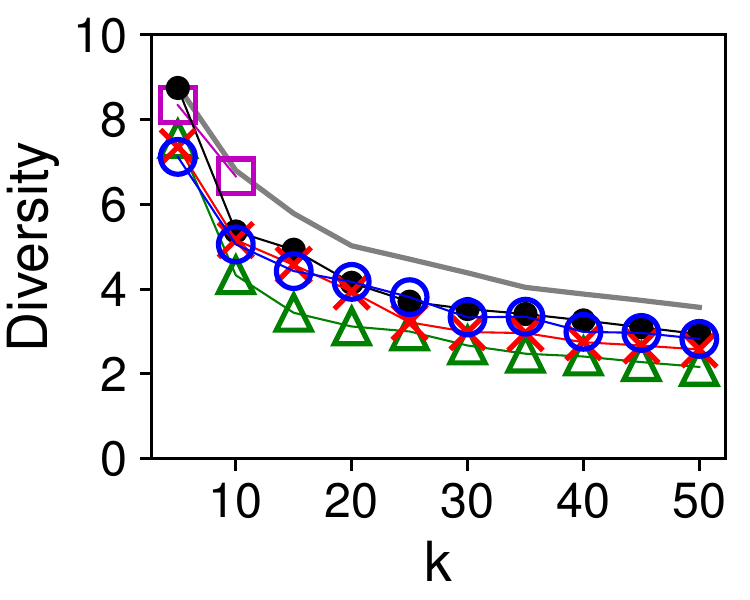}
  }
  \hspace{1em}
  \subfigure[CelebA (Age, $m=2$)]{
    \includegraphics[width=0.2\textwidth]{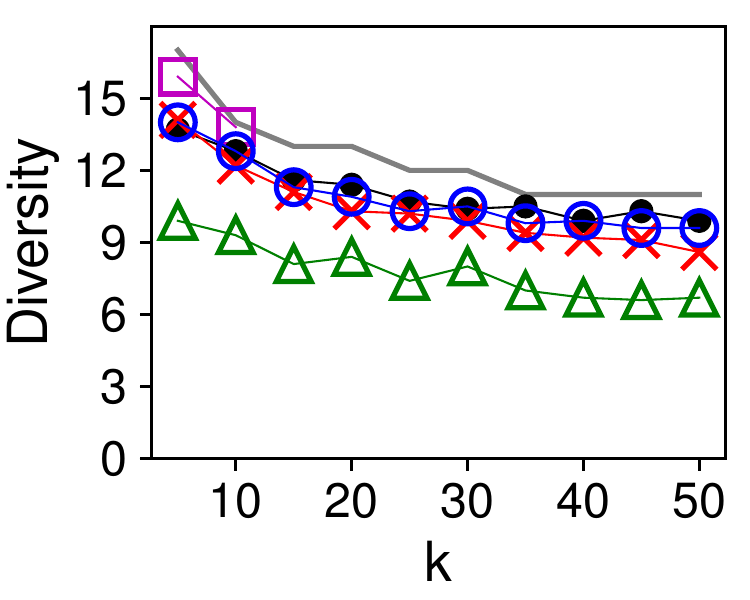}
  }
  \hspace{1em}
  \subfigure[CelebA (Sex, $m=2$)]{
    \includegraphics[width=0.2\textwidth]{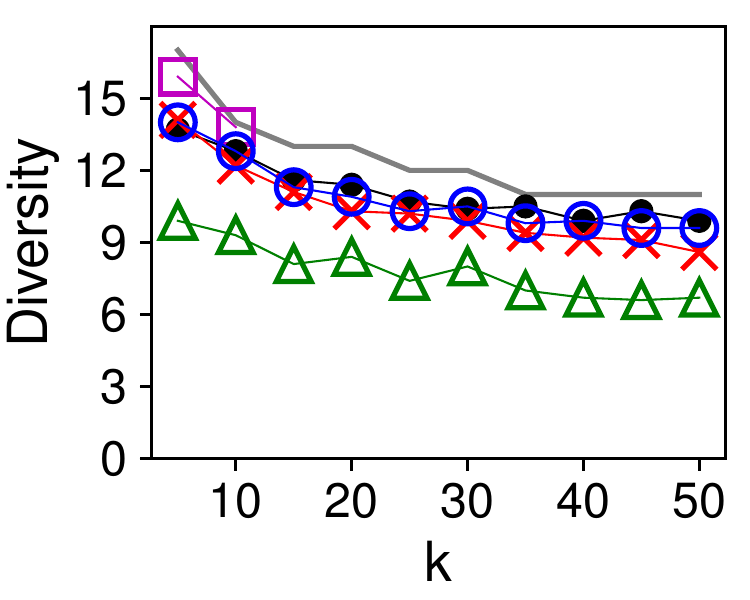}
  }
  \hspace{1em}
  \subfigure[Census (Sex, $m=2$)]{
    \includegraphics[width=0.2\textwidth]{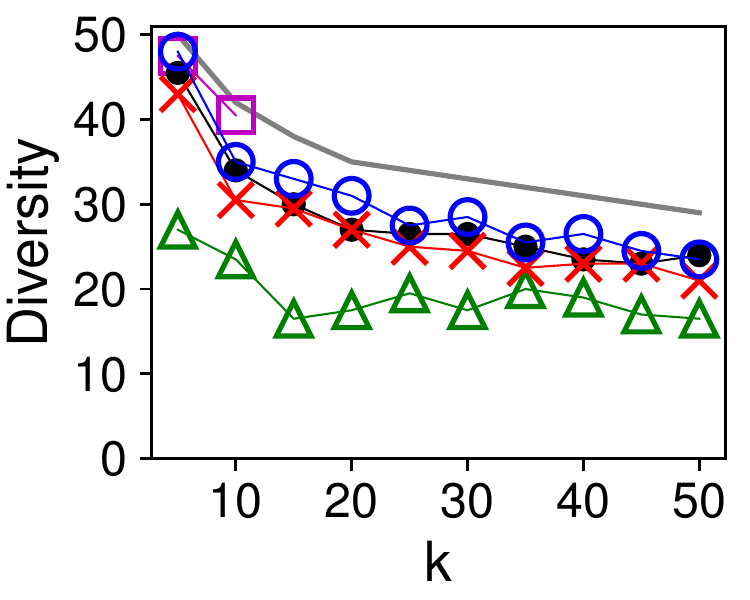}
  }
  \\
  \subfigure[Adult (Race, $m=5$)]{
    \includegraphics[width=0.2\textwidth]{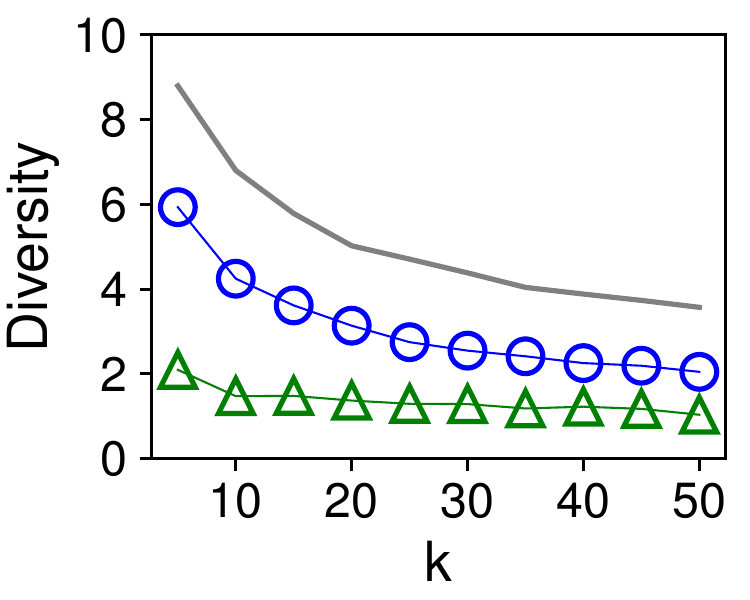}
  }
  \hspace{1em}
  \subfigure[CelebA (Sex+Age, $m=4$)]{
    \includegraphics[width=0.2\textwidth]{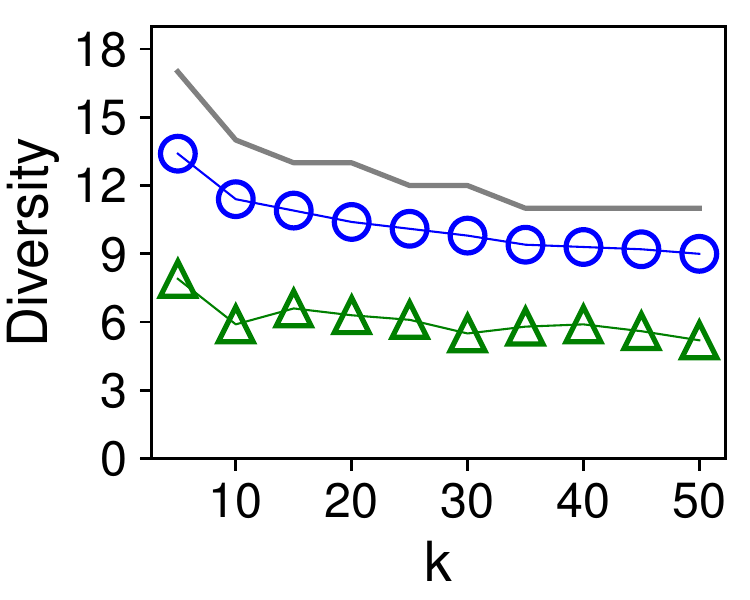}
  }
  \hspace{1em}
  \subfigure[Census (Age, $m=7$)]{
    \includegraphics[width=0.2\textwidth]{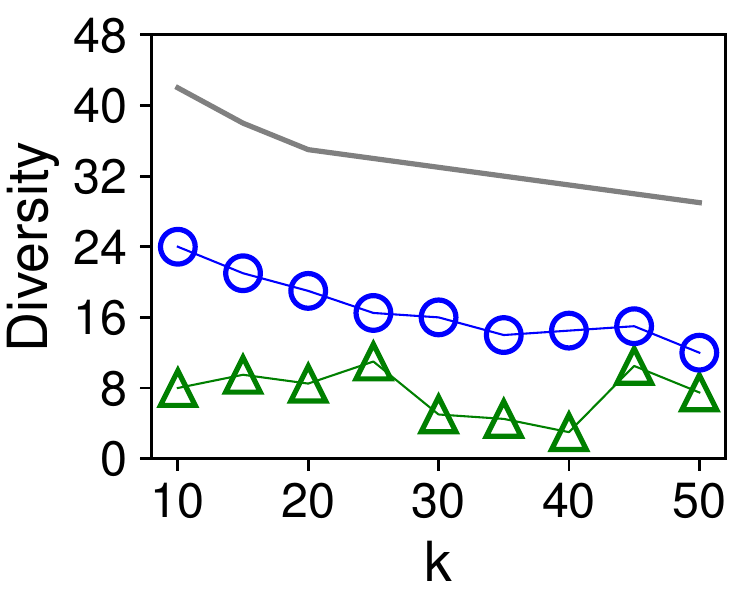}
  }
  \hspace{1em}
  \subfigure[Lyrics (Genre, $m=15$)]{
    \includegraphics[width=0.2\textwidth]{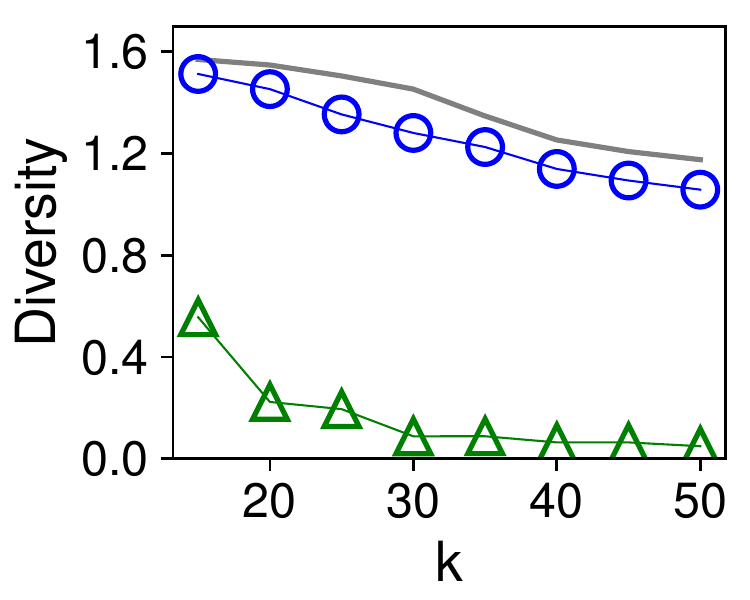}
  }
  \caption{Solution quality with varying $k$. The diversity values of \textsf{GMM} are plotted as gray lines to illustrate the losses caused by fairness constraints.}
  \label{fig:exp:k:m2:div}
\end{figure*}

\begin{figure*}
  \centering
  \includegraphics[width=0.5\textwidth]{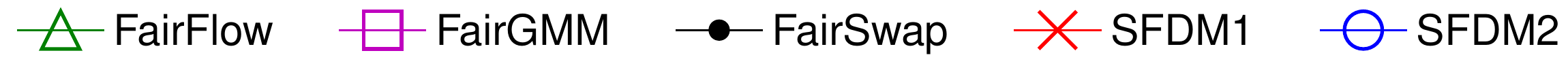}
  \\
  \subfigure[Adult (Sex, $m=2$)]{
    \includegraphics[width=0.2\textwidth]{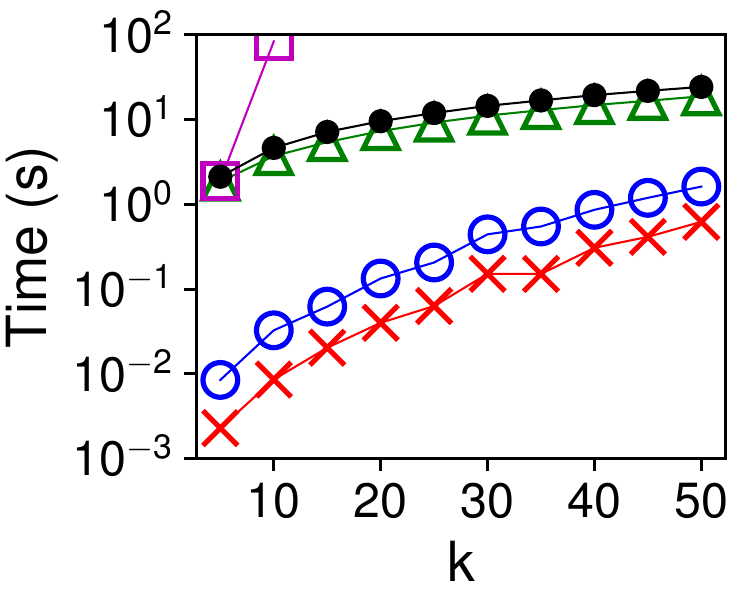}
  }
  \hspace{1em}
  \subfigure[CelebA (Age, $m=2$)]{
    \includegraphics[width=0.2\textwidth]{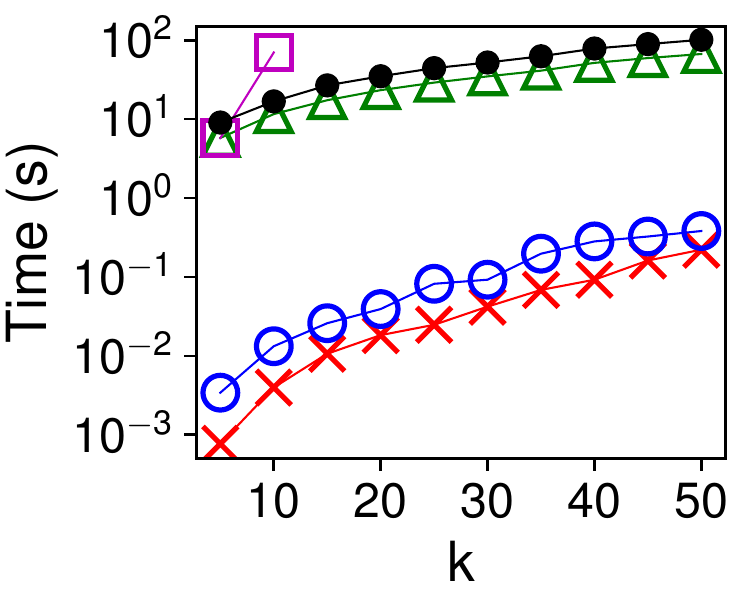}
  }
  \hspace{1em}
  \subfigure[CelebA (Sex, $m=2$)]{
    \includegraphics[width=0.2\textwidth]{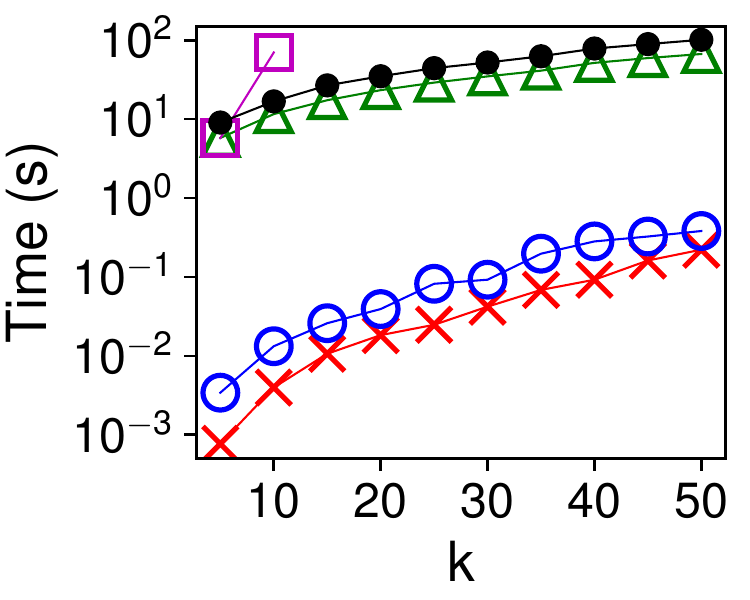}
  }
  \hspace{1em}
  \subfigure[Census (Sex, $m=2$)]{
    \includegraphics[width=0.2\textwidth]{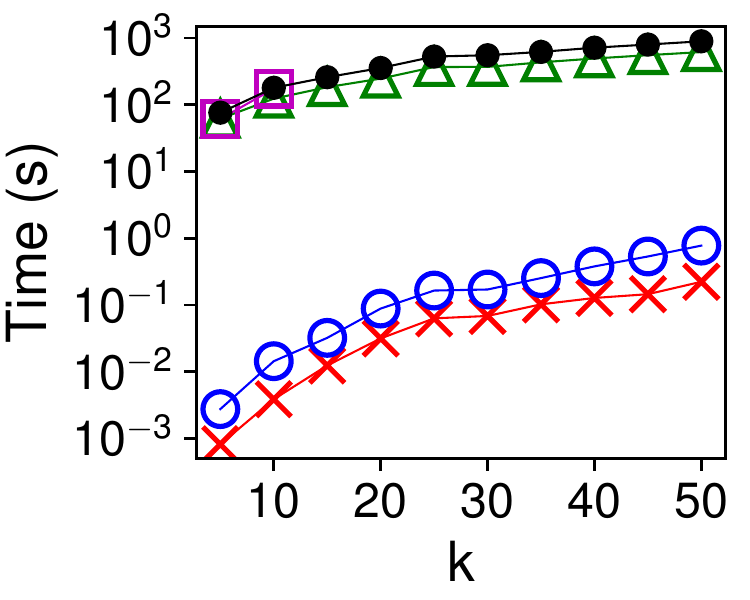}
  }
  \\
  \subfigure[Adult (Race, $m=5$)]{
    \includegraphics[width=0.2\textwidth]{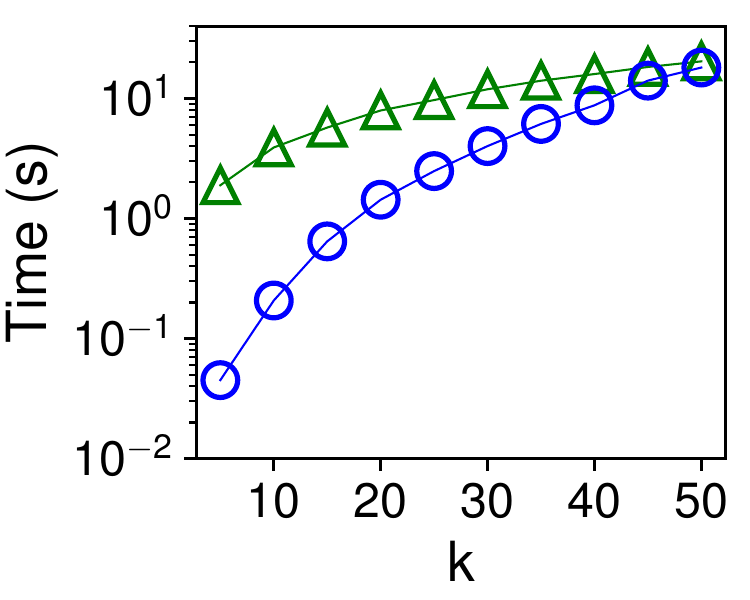}
  }
  \hspace{1em}
  \subfigure[CelebA (Sex+Age, $m=4$)]{
    \includegraphics[width=0.2\textwidth]{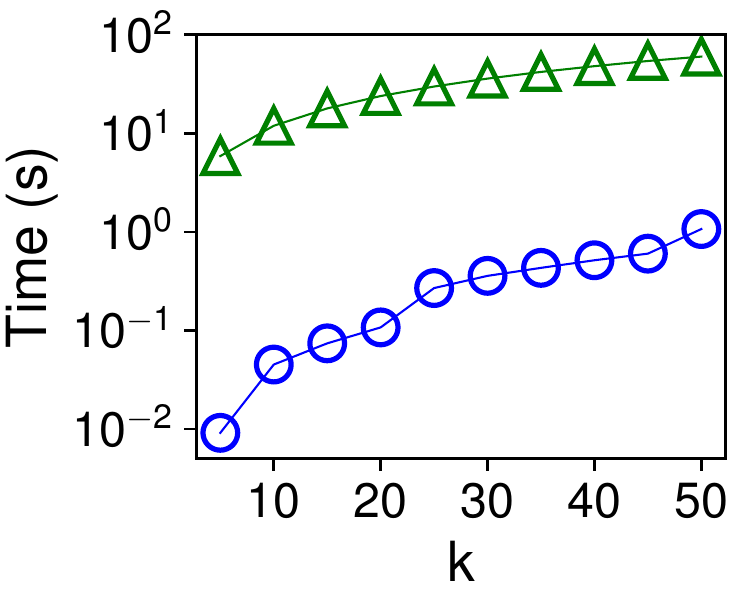}
  }
  \hspace{1em}
  \subfigure[Census (Age, $m=7$)]{
    \includegraphics[width=0.2\textwidth]{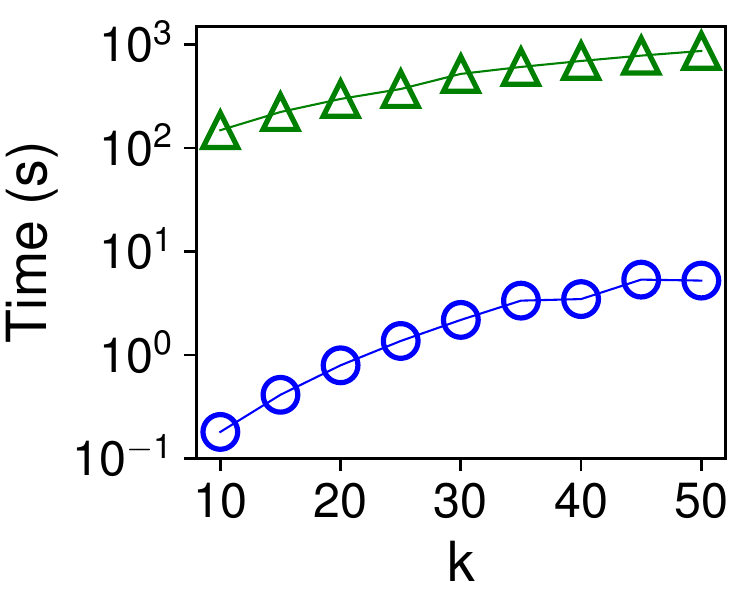}
  }
  \hspace{1em}
  \subfigure[Lyrics (Genre, $m=15$)]{
    \includegraphics[width=0.2\textwidth]{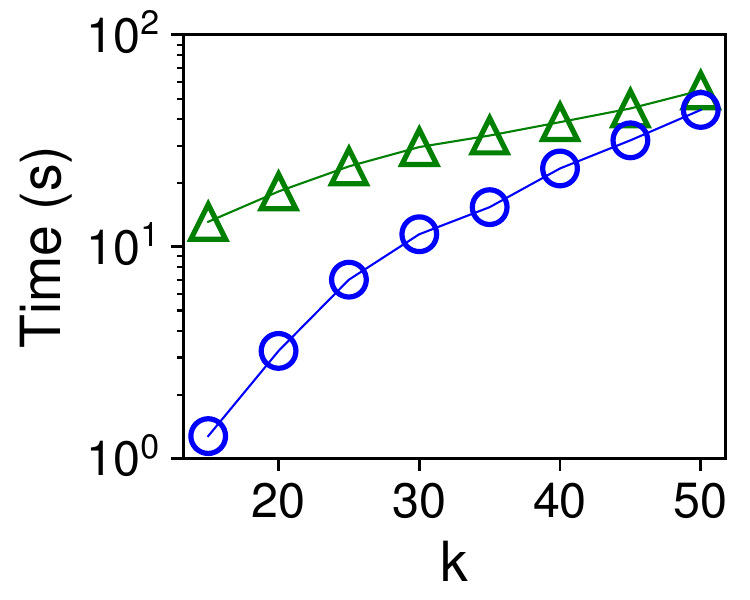}
  }
  \caption{Efficiency with varying $k$.}
  \label{fig:exp:k:m2:time}
\end{figure*}

\subsection{Experimental Results}
\label{subsec:results}

\noindent\textbf{Overview:}
Table~\ref{tbl:exp:overview} shows the performance of different algorithms for FDM on four real-world datasets with different group settings when the solution size $k$ is fixed to $20$. Note that \textsf{FairGMM} is not included in Table~\ref{tbl:exp:overview} because it needs to enumerate at most $\binom{km}{k}=O(m^k)$ candidates for solution computation and cannot scale to $k>10$ and $m>5$. First of all, compared with the unconstrained solution returned by \textsf{GMM}, all the fair solutions are less diverse because of additional fairness constraints. Since \textsf{GMM} is a $\frac{1}{2}$-approximation algorithm and $\mathtt{OPT} \geq \mathtt{OPT}_f$, $2 \cdot div\_\mathsf{GMM}$ can be seen as an upper bound of $\mathtt{OPT}_f$, from which we can find that all four algorithms return solutions of much better approximations than the lower bounds.

In case of $m=2$, \textsf{SFDM1} runs the fastest among all four algorithms, which achieves two to four orders of magnitude speedups over \textsf{FairSwap} and \textsf{FairFlow}. Meanwhile, its solution quality is close or equal to that of \textsf{FairSwap} in most cases. \textsf{SFDM2} shows lower efficiency than \textsf{SFDM1} due to higher cost of post-processing. But it is still much more efficient than offline algorithms by taking the advantage of stream processing. In addition, the solution quality of \textsf{SFDM2} benefits from the greedy selection procedure in Algorithm~\ref{alg:matroid}, which is not only consistently better than that of \textsf{SFDM1} but also better than that of \textsf{FairSwap} on \emph{Adult} and \emph{Census}.

In case of $m>2$, \textsf{SFDM1} and \textsf{FairSwap} are not applicable any more and thus ignored in Table~\ref{tbl:exp:overview}. \textsf{SFDM2} shows significant advantages over \textsf{FairFlow} in terms of both solution quality and efficiency. It provides solutions of up to $6.3$ times more diverse than \textsf{FairFlow} while running several orders of magnitude faster.

In terms of space usage, both \textsf{SFDM1} and \textsf{SFDM2} store very small portions of elements (less than $0.1\%$ on \emph{Census}) on all datasets. \textsf{SFDM2} keeps slightly more elements than \textsf{SFDM1} because the capacity of each group-specific candidate for group $i$ is $k$ instead of $k_i$. For \textsf{SFDM2}, the number of stored elements increases near linearly with $m$, since the total number of candidates is linear to $m$.

\noindent\textbf{Effect of Parameter $\varepsilon$:}
Fig.~\ref{fig:exp:eps} illustrates the performance of \textsf{SFDM1} and \textsf{SFDM2} with different values of $\varepsilon$ when $k$ is fixed to $20$. We range the value of $\varepsilon$ from $0.05$ to $0.25$ on \emph{Adult}, \emph{CelebA}, and \emph{Census} and from $0.02$ to $0.1$ on \emph{Lyrics}. Since the angular distances between two vectors in \emph{Lyrics} are at most $\frac{\pi}{2}$, too large values of $\varepsilon$ will lead to great estimation errors for $\mathtt{OPT}_f$. Generally, \textsf{SFDM1} has higher efficiency and smaller space usage than \textsf{SFDM2} for different values of $\varepsilon$, but \textsf{SFDM2} exhibits better solution quality. Furthermore, the running time and numbers of stored elements of both algorithms significantly decrease when the value of $\varepsilon$ increases. This is consistent with our analyses in Section~\ref{sec:alg} because the number of guesses for $\mathtt{OPT}_f$ and thus the number of candidates maintained by both algorithms are $O(\frac{\log \Delta}{\varepsilon})$. A slightly surprising result is that the diversity values of the solutions do not degrade obviously even when $\varepsilon=0.25$. This can be explained by the fact that both algorithms return the best solutions after post-processing among all candidates, which means that they can provide good solutions as long as there is some $\mu \in \mathcal{U}$ close to $\mathtt{OPT}_f$. We infer that such $\mu$ still exists when $\varepsilon=0.25$. Nevertheless, we note that the chance of finding an appropriate value of $\mu$ will be smaller when the value of $\varepsilon$ is larger, which will lead to less stable solution quality. In the remaining experiments, we always use $\varepsilon=0.1$ for both algorithms on all datasets except \emph{Lyrics}, where the value of $\varepsilon$ is set to $0.05$.

\begin{figure}
  \centering
  \subfigure[Adult]{
    \includegraphics[width=0.22\textwidth]{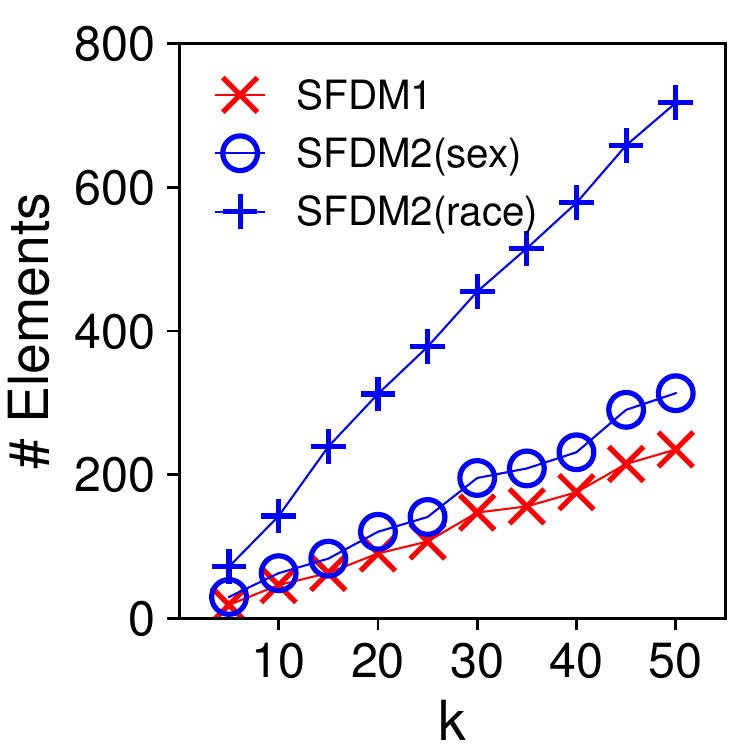}
  }
  \subfigure[Census]{
    \includegraphics[width=0.22\textwidth]{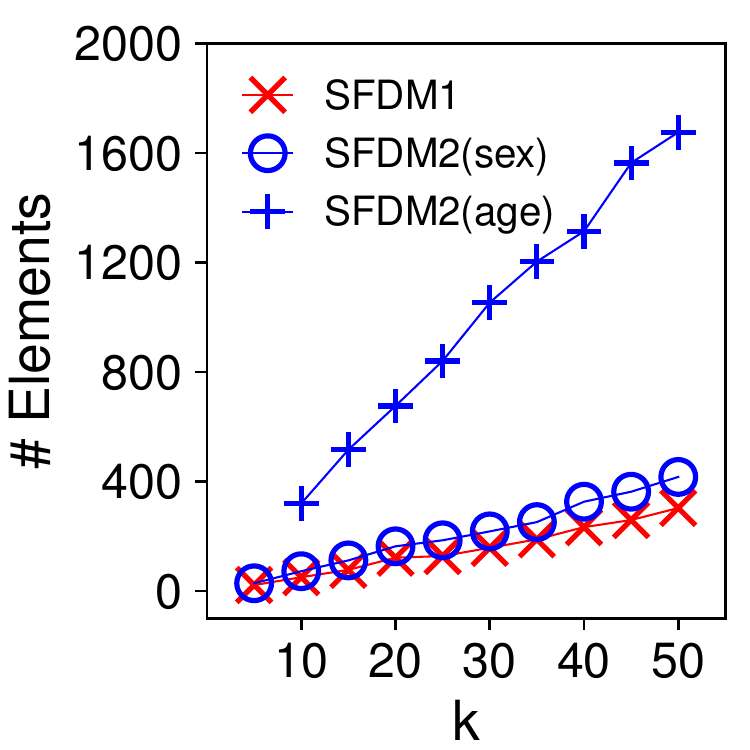}
  }
  \caption{Number of stored elements with varying $k$.}
  \label{fig:exp:k:elem}
\end{figure}

\vspace{1mm}
\noindent\textbf{Effect of Solution Size $k$:}
The impact of $k$ on the performance of different algorithms is illustrated in Fig.~\ref{fig:exp:k:m2:div}--\ref{fig:exp:k:elem}. Here we vary $k$ in $[5,50]$ when $m \leq 5$, or $[10,50]$ when $5 < m \leq 10$, or $[15,50]$ when $m > 10$, since we restrict that an algorithm must pick at least one element from each group. In general, for each algorithm, the diversity value drops with $k$ as the diversity function is monotonically non-increasing while the running time grows with $k$ as their time complexities are linear or quadratic w.r.t.~$k$. Compared with the solutions of \textsf{GMM}, all fair solutions are slightly less diverse when $m=2$. The gaps in diversity values become more obvious when $m$ is larger. Although \textsf{FairGMM} achieves slightly higher solution quality than other algorithms when $k \leq 10$ and $m=2$, it is not scalable to larger $k$ and $m$ due to the huge cost of enumeration. The solution quality of \textsf{FairSwap}, \textsf{SFDM1}, and \textsf{SFDM2} is close to each other when $m=2$, which is better than that of \textsf{FairFlow}. But the efficiencies of \textsf{SFDM1} and \textsf{SFDM2} are orders of magnitude higher than those of \textsf{FairSwap} and \textsf{FairFlow} when $m=2$. Furthermore, when $m > 2$, \textsf{SFDM2} outperforms \textsf{FairFlow} in terms of both efficiency and effectiveness across all $k$ values. However, since the time complexity of \textsf{SFDM2} is quadratic w.r.t.~both $k$ and $m$, its running time increases drastically with $k$ when $m$ is large. Finally, in terms of space usage, the numbers of elements maintained by \textsf{SFDM1} and \textsf{SFDM2} both increase linearly with $k$. In addition, it is also linearly correlated with $m$ for \textsf{SFDM2}. In all experiments, both algorithms only store small portions of elements in the dataset.

\begin{figure}
  \centering
  \subfigure[Adult (Sex, $k=20$)]{
    \includegraphics[width=0.215\textwidth]{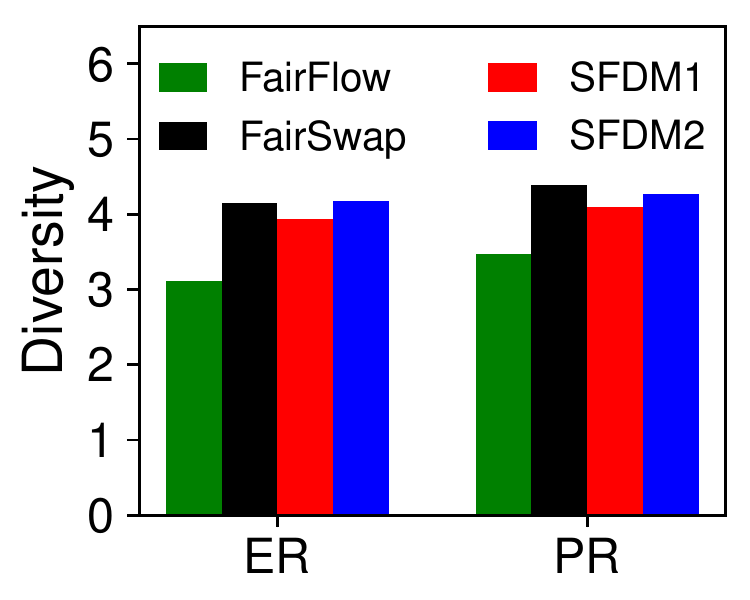}
    \hspace{0.5em}
    \includegraphics[width=0.215\textwidth]{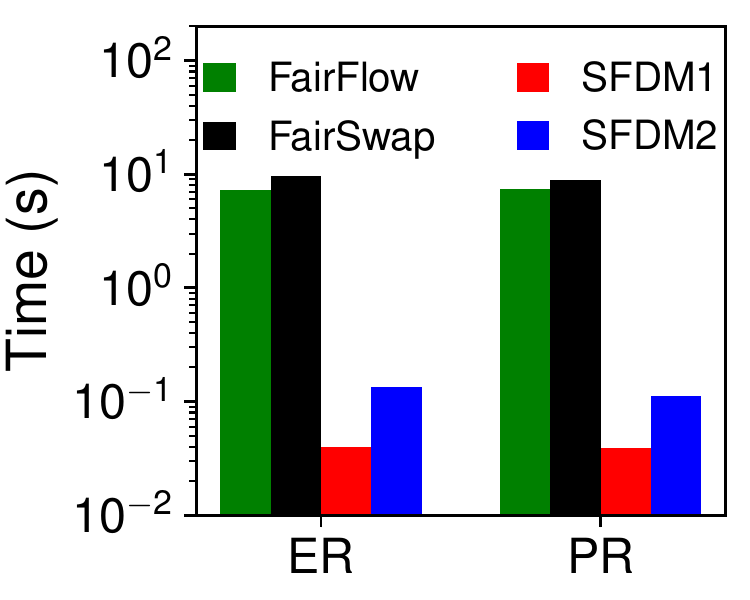}
  }
  \subfigure[Adult (Race, $k=20$)]{
    \includegraphics[width=0.215\textwidth]{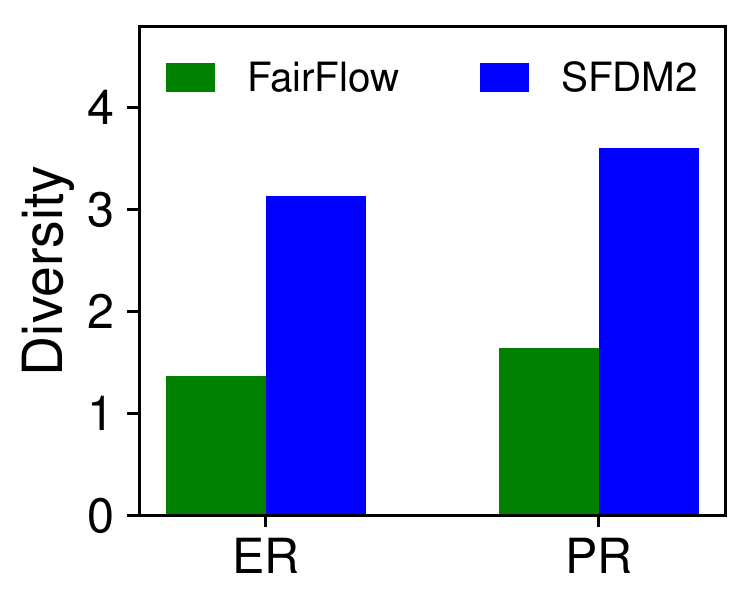}
    \hspace{0.5em}
    \includegraphics[width=0.215\textwidth]{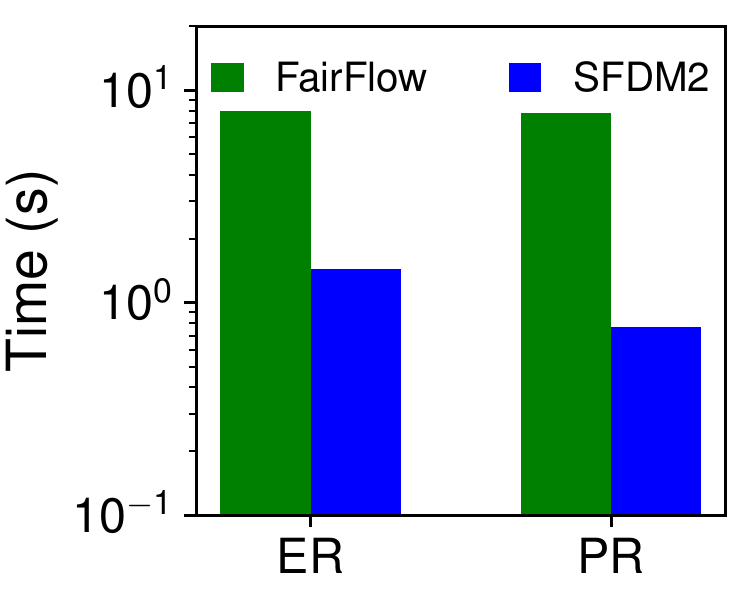}
  }
  \caption{Comparison of the performance of different algorithms on \emph{Adult} for equal representation (ER) and proportional representation (PR).}
  \label{fig:exp:er:pr}
\end{figure}

\begin{figure*}
  \centering
  \includegraphics[width=0.4\textwidth]{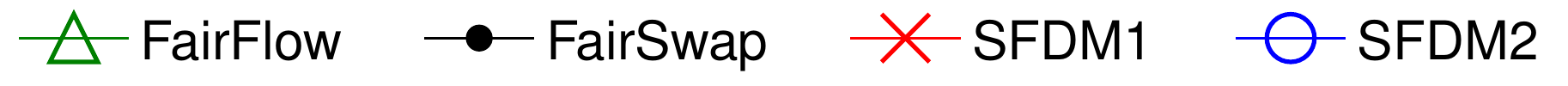}
  \\
  \subfigure[Synthetic ($m=2$)]{
    \includegraphics[width=0.2\textwidth]{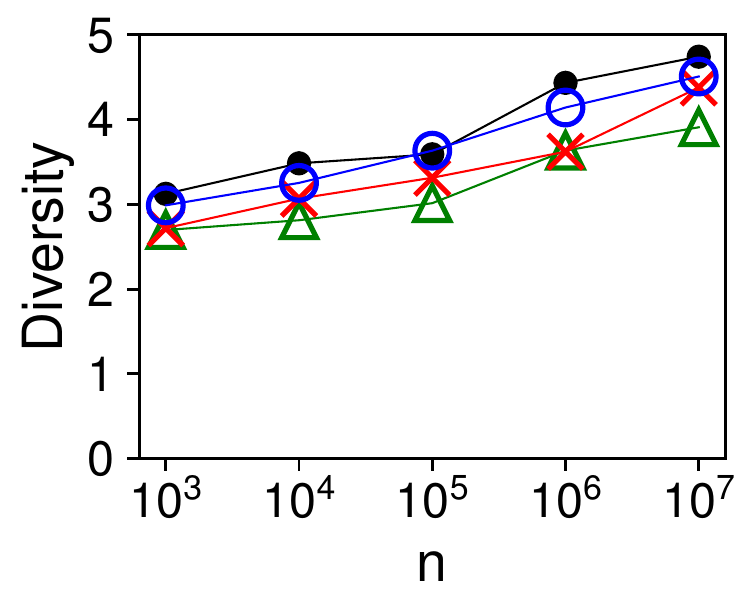}
    \includegraphics[width=0.2\textwidth]{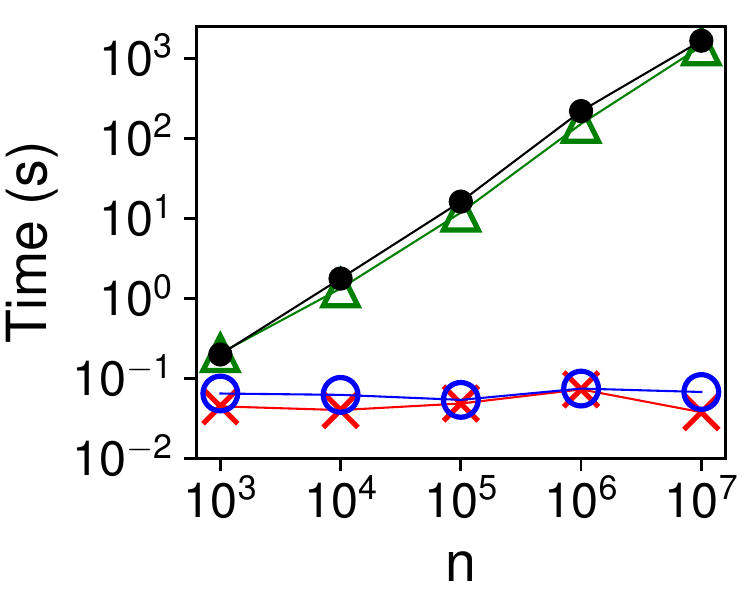}
  }
  \hspace{1em}
  \subfigure[Synthetic ($m=10$)]{
    \includegraphics[width=0.2\textwidth]{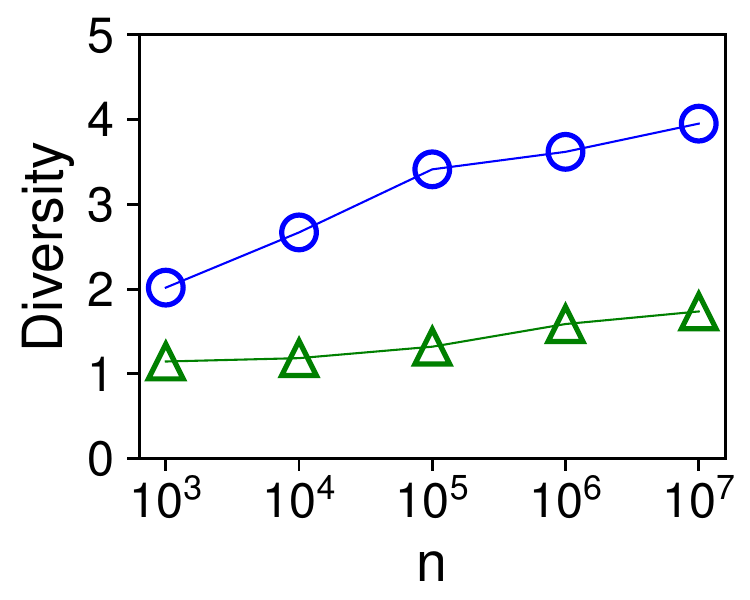}
    \includegraphics[width=0.2\textwidth]{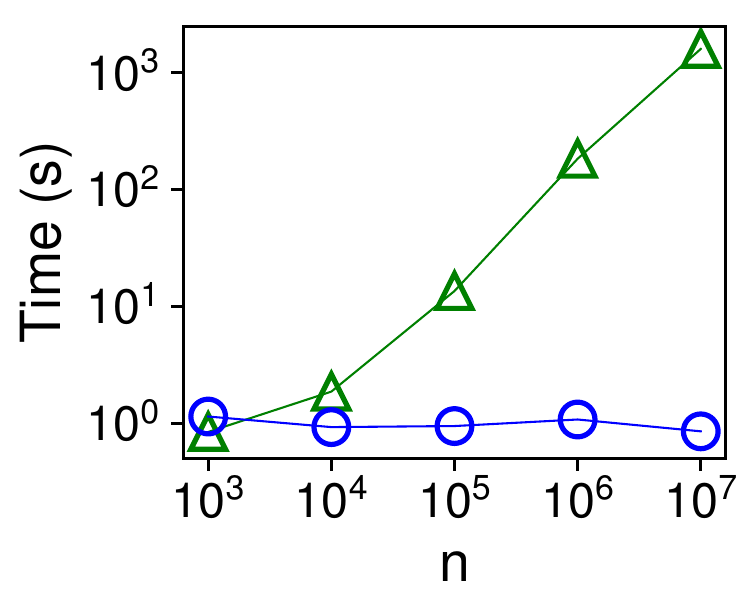}
  }
  \caption{Solution quality and running time with varying $n$ ($k=20$).}
  \label{fig:exp:n}
\end{figure*}

\noindent\textbf{Equal Representation (ER) vs.~Proportional Representation (PR):}
Fig.~\ref{fig:exp:er:pr} compares the solution quality and running time of different algorithms for two popular notions of \emph{fairness} -- i.e., \emph{equal representation} (ER) and \emph{proportional representation} (PR), when $k=20$ on \emph{Adult} with highly skewed groups, where $67\%$ of the records are for males and $87\%$ of the records are for Whites. The diversity value of the solution of each algorithm is slightly higher for PR than ER, as the solution for PR is closer to the unconstrained one. The running time of \textsf{SFDM1} and \textsf{SFDM2} is slightly shorter for PR than ER since fewer swapping and augmentation steps are performed on each candidate in the post-processing.

\begin{figure}
  \centering
  \includegraphics[width=0.4\textwidth]{exp-mn/legend-mn.pdf}
  \\
  \includegraphics[width=0.2\textwidth]{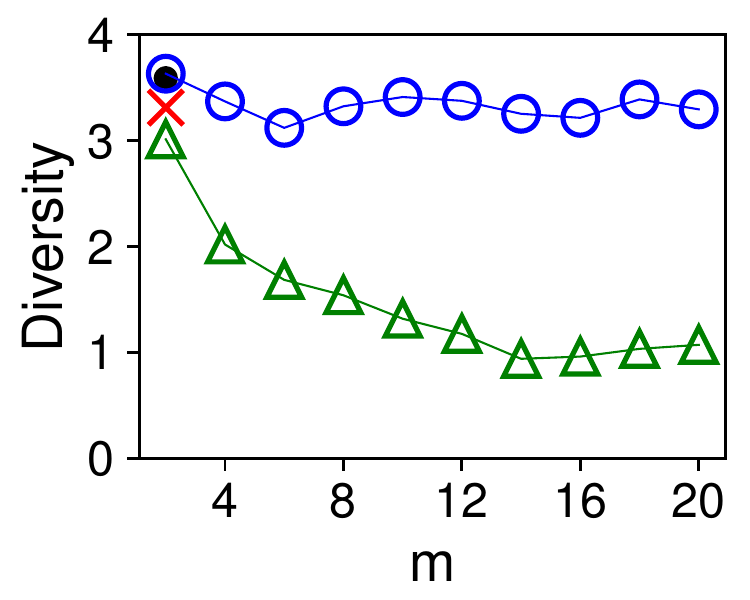}
  \hspace{1em}
  \includegraphics[width=0.2\textwidth]{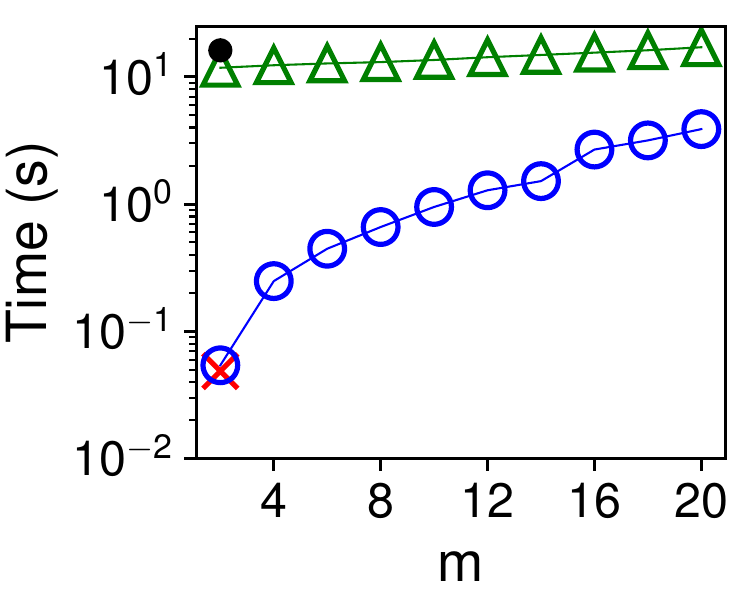}
  \caption{Solution quality and running time with varying $m$ on the synthetic datasets ($n=10^5, k=20$).}
  \label{fig:exp:m}
\end{figure}

\noindent\textbf{Scalability:}
We evaluate the scalability of each algorithm on synthetic datasets with varying the number of groups $m$ from $2$ to $20$ and the dataset size $n$ from $10^3$ to $10^7$. The results on solution quality and running time for different values of $n$ and $m$ when $k=20$ are presented in Fig.~\ref{fig:exp:n} and~\ref{fig:exp:m}, respectively. We omit the results on space usages because they are similar to previous results. First of all, \textsf{SFDM2} shows much better scalability than \textsf{FairFlow} w.r.t.~$m$ in terms of solution quality. The diversity value of the solution \textsf{SFDM2} only slightly deceases with $m$ and is up to $3$ times higher than that of \textsf{FairFlow} when $m>10$. Nevertheless, its running time increases more rapidly with $m$ due to the quadratic dependence on $m$. Furthermore, the diversity values of different algorithms slightly grow with $n$ but are always close to each other for different values of $n$ when $m=2$. When $m=10$, the advantage of \textsf{SFDM2} over \textsf{FairFlow} in solution quality becomes larger with $n$. Finally, the running time (as well as memory usage) of offline algorithms are linear to $n$. But the running time and memory usage of \textsf{SFDM1} and \textsf{SFDM2} are independent of $n$, as analyzed in Section~\ref{sec:alg}.

\section{Conclusion}
\label{sec:conclusion}

In this paper, we studied the diversity maximization problem with fairness constraints in the streaming setting. First of all, we proposed a $\frac{1-\varepsilon}{4}$-approximation streaming algorithm for this problem when there were two groups in the dataset. Furthermore, we designed a $\frac{1-\varepsilon}{3m+2}$-approximation streaming algorithm that could deal with an arbitrary number $m$ of groups in the dataset. Extensive experiments on real-world and synthetic datasets confirmed the efficiency, effectiveness, and scalability of our proposed algorithms.

In future work, we would like to improve the approximation ratios of the proposed algorithms and consider diversity maximization problems with fairness constraints in more general settings, e.g., the sliding-window model and fairness constraints defined on multiple sensitive attributes.

\section*{Acknowledgments}
This research was done when Yanhao Wang worked at the University of Helsinki. Yanhao Wang and Michael Mathioudakis have been supported by the MLDB project of the Academy of Finland (decision number: 322046). Francesco Fabbri is a fellow of Eurecat's ``Vicente L\'{o}pez'' PhD grant program. This work was partially financially supported by the Catalan Government through the funding grant ACCI\'{O}-Eurecat (Project PRIVany-nom).

\balance
\bibliography{references}

\end{document}